\documentclass[journal]{IEEEtran}

\usepackage[monochrome]{color}
\usepackage[T1]{fontenc}
\usepackage{graphicx}
\usepackage{cite}
\usepackage{subcaption}
\usepackage{overpic}
\usepackage{amssymb}
\usepackage[cmex10]{amsmath}
\usepackage{amsthm}
\usepackage{enumitem}
\usepackage{todonotes}
\usepackage{stfloats}
\usepackage{bbm}
\usepackage{siunitx}

\usepackage{mathtools}

\IEEEoverridecommandlockouts

\def\diag{\mathrm{diag}}

\def\Htran{\mbox{\tiny $\mathrm{H}$}}
\def\Ttran{\mbox{\tiny $\mathrm{T}$}}
\def\CN{\mathcal{N}_{\mathbb{C}}} 
\def\imagunit{\mathsf{j}} 

\newcommand{\Exp}{{\mathbb{E}}}
\newcommand{\re}{{\mathrm{Re}}}
\newcommand{\imag}{{\mathrm{Im}}}

\theoremstyle{plain}

\newtheorem{lemma}{Lemma}
\newtheorem{corollary}{Corollary}

\newtheorem{remark}{Remark}
\newtheorem{assumption}{Assumption}

\newcommand{\vect}[1]{{\bf{#1}}}

\newcommand{\ZMT}{{${\bf Z}_{\rm MT}$}}
\newcommand{\ZMR}{{${\bf Z}_{\rm MR}$}}

\begin{document}

\title{\vspace{-0.7cm}\huge{Holographic MIMO Communications: \\ What is the benefit of closely spaced antennas?}}

\author{
\IEEEauthorblockN{Antonio Alberto D'Amico, Luca Sanguinetti, \emph{Senior Member, IEEE}\vspace{-1.cm}
\thanks{
\newline \indent A. A. D'Amico and L.~Sanguinetti are with the Dipartimento di Ingegneria dell'Informazione, University of Pisa, Pisa, Italy (email: luca.sanguinetti@unipi.it, antonio.damico@unipi.it). This work is supported by Huawei Technologies Sweden AB.
}
}}
\maketitle

\begin{abstract}\vspace{0cm}
Holographic MIMO refers to an array (possibly large) with a massive number of antennas that are individually controlled and densely deployed. The aim of this paper is to provide further insights into the advantages (if any) of having closely spaced antennas in the uplink and downlink of a multi-user Holographic MIMO system. To this end, we make use of the multiport communication theory, which ensures physically consistent uplink and downlink models. We first consider a simple uplink scenario with two side-by-side half-wavelength dipoles, two users and single path line-of-sight propagation, and show both analytically and numerically that the \textcolor{blue}{array gain} and average spectral efficiency depend strongly on the directions from which the signals are received and on the array matching network used. The numerical results are subsequently leveraged to extend the analysis into more practical scenarios involving larger arrays of dipoles \textcolor{blue}{(arranged in a uniform linear array)} and a greater number of users. \textcolor{red}{The case where the antennas are densely packed in a space-constrained factor form is also considered. It is found that the spectral efficiency increases with decreasing antenna spacing only for arrays of moderate size, e.g. in the order of a few wavelengths. Comparatively, larger arrays exhibit only marginal improvements in spectral efficiency when compared to arrays with half-wavelength spacing.}
\end{abstract}

\begin{IEEEkeywords}\vspace{0cm}
Holographic MIMO, mutual coupling, circuit theory, matching networks, uplink/downlink duality.
\end{IEEEkeywords}

\section{Introduction} 
Communication theorists are always on the lookout for new technologies to improve the speed and reliability of wireless communications. Chief among the technologies that blossomed into major advances is the multiple antenna technology, whose latest implementation is Massive MIMO (multiple-input multiple-output)~\cite{marzetta2010noncooperative, massivemimobook}. Inspired by its potential benefits~\cite{sanguinettiTCOM2020}, new research directions are taking place under different names~\cite{BJORNSON20193}, e.g., Holographic MIMO~\cite{Huang2020} and large intelligent surfaces~\cite{Rusek2018}. Particularly, the former concept refers to an array (possibly \emph{electromagnetically large}, i.e., compared to the wavelength) with a massive number of closely spaced antennas whose electromagnetic interactions inevitably results into mutual coupling~\cite{balanis}. Although exceptions exist, e.g., \cite{7831497,6843218,8350292,9048753,9838533,Heath_2023a,10158708}, the vast majority of the MIMO literature has entirely neglected mutual coupling since it is all about using (possibly \emph{physically large}) arrays with \emph{half-wavelength} antenna spacing~\cite{massivemimobook}. Another major caveat of the classical MIMO literature (in general) is that it mostly relies on the abstractions of signal processing and information theories, which are not always consistent with the physical context of the underlying system.  
Fortunately, there exists a thin, but solid, literature that can be used to overcome these limitations~\cite{Janaswamy2002,Svantesson_ICASSP2001,Wallace2004,Nossek2010,Nossek2014} but its development has been relatively slow due to the less tractable analysis.

 The first attempts in this direction can be found in~\cite{Janaswamy2002,Svantesson_ICASSP2001,Wallace2004}. Particularly, in~\cite{Wallace2004} the authors derived the model of a single-user MIMO communication system as an electrical network described by scattering matrices. This allows to account for the mutual coupling between transmit and/or receive antennas. \textcolor{blue}{A matching network was also introduced at the receiver to maximize the power transfer from the outputs of the receive antennas to the loads.} The framework developed in~\cite{Wallace2004} is also among the first to connect the physical power to the abstract concept used in signal and information theories. An alternative framework is developed in~\cite{Nossek2010,Nossek2014} based on the multiport communication theory. This involves a circuit theoretic approach where the inputs and outputs of the multiple antenna communication system are associated with ports of a multi-port black-box, described by impedance matrices. \textcolor{blue}{Notice that the two frameworks above are equivalent} and the multiport communication theory has been used in the MIMO literature to study several aspects. 
 {For example, in~\cite{yordanov2009arrays,Ivrlac2009ICC,Laas2020} the transmit/receive array gain is evaluated (with and without matching networks) for uniform linear and circular arrays. The diversity gain is investigated in \cite{ivrlavc2011diversity}, while the effects of the antenna separation on the mutual information of two Hertzian dipoles are analyzed in~\cite{Nossek2014}. The multiport communication theory is also used in~\cite{Laas2020_Reciprocity} for studying the uplink/downlink reciprocity and mutual information of multi-user MIMO systems. More recently, \cite{Bamelak_2023} used it to investigate the impact of mutual coupling in the channel estimation of single-user MIMO communications.}    

 The main objectives of this paper are two fold: $i$) to use the multiport communication theory to derive physically consistent uplink and downlink models for multi-user Holographic MIMO communications with linear processing; and $ii$) to use the developed models to answer the following question: \textit{what are the spectral efficiency advantages (if any) of having closely spaced antennas?} To answer this question, we first consider a simple uplink scenario with two side-by-side half-wavelength dipoles at the base station (BS), two user equipments (UEs) and single path line-of-sight (LoS) propagation. In this context, we show both analytically and numerically that the \textcolor{blue}{array gain}, interference gain and spectral efficiency depend strongly on
the directions from which the UE signals are received and on the array matching network used at the BS. \textcolor{blue}{Benefits are only attainable through impedance matching (e.g., \cite{Volodymyr_2022,Heath_2023a,10158708}) and specific incident signal directions. However, implementing impedance matching for arrays with numerous antennas poses significant challenges \cite{Laas2020_Reciprocity}. Furthermore, in practical multi-user systems, signal directions are uncontrollable due to their dependency on user locations.} In these cases, the gains may be marginal or even non-existent.
The internal losses within the dipole antennas are also shown to significantly impact the
spectral efficiency as the spacing reduces. Numerical results are then used to show that similar conclusions
hold in more practical scenarios with any number of UEs and any number of \textcolor{blue}{side-by-side dipole antennas arranged in a uniform linear array at the BS.} Particularly, the analysis is conducted in the following two cases: $i$) the number of dipoles is fixed as we vary their spacing; \textcolor{blue}{$ii$)} \textcolor{blue}{the size of the uniform linear array} is fixed as we vary the dipole spacing. In the latter case, it turns out that the spectral efficiency increases \textcolor{blue}{as} the antenna distance reduces.  \textcolor{red}{This is most noticeable for arrays spanning a few wavelengths. In contrast, larger arrays show only marginal improvements in spectral efficiency compared to half-wavelength arrays.}  

Although most of the analysis focuses on the uplink, we also investigate the downlink. Particular attention is given to the uplink and downlink duality in the presence of different matching networks. {Specifically, we show that the downlink and uplink channels are reciprocal up to a linear transformation. In line with \cite{Laas2020_Reciprocity}, the ordinary channel reciprocity (i.e., no linear transformation) holds true only if full matching networks (that are hard to implement in arrays with many antennas \cite{Laas2020_Reciprocity}) are employed at both sides.} Numerical results are used to quantify the spectral efficiency loss when the linear transformation is not applied.

The remainder of this paper is organized as follows. In Section~\ref{sec:system_model}, we review the Multiport Communication Theory from~\cite{Nossek2014}. In Section~III, we show how to compute the impedance matrices when a uniform linear array made of half-wavelength dipoles is used at both sides. In Section IV, the uplink and downlink signal models for Holographic MIMO communications are derived on the basis of the
multiport communication model. The concept of uplink and downlink duality is also discussed. To showcase what is the impact of mutual coupling, a simple case study with two dipole antennas and two UEs is considered in Section V. The analysis is then extended in Section VI to more realistic scenarios with multiple antennas, multiple UEs and arrays of varying or fixed aperture. Conclusions are drawn in Section VI. 

\textcolor{blue}{\textit{Notation:} Lower-case bold letters are used for vectors and upper-case bold letters are used for matrices. ${\bf n}\sim \mathcal{N}_\mathbb{C}({\bf 0}, {\bf R})$ denotes the circularly symmetric complex Gaussian distribution with zero mean and covariance matrix ${\bf R}$. We use $\mathbb{E}\{\cdot\}$ to indicate the expectation operator. The operators $^{\Ttran}$, $^{*}$ , and $^{\Htran}$ denote transpose, complex conjugate, and Hermitian transpose, respectively. The Euclidean norm is denoted by $\left\|\cdot\right\|$ and $|\cdot|$ is the absolute value. We use ${\bf a} \cdot {\bf b}$ and $\odot$ to denote the scalar product and the Hadamard product between ${\bf a}$ and ${\bf b}$, respectively.}

\textit{Reproducible research:} The Matlab code used to obtain the simulation results will be made available upon completion of the review process.


\begin{figure}[t!]\vspace{0cm}
\centering
\includegraphics[width = 1\columnwidth]{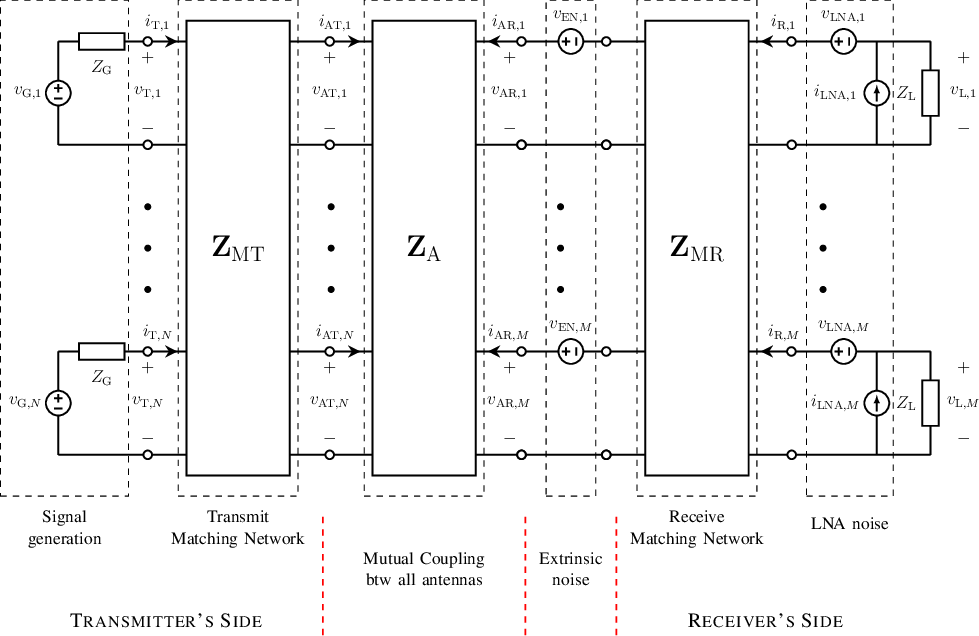}\vspace{-0.1cm}
\caption{{\textcolor{blue}{Physical model of a multi-antenna communication system, based on the circuit theoretic concept of linear multiports~\cite{Nossek2014}}.}
\vspace{-0.25cm}}
\label{fig: Network Model}
\end{figure}

\section{Review of Multiport Communication Theory}  \label{sec:system_model}

Consider a narrowband communication system equipped with $M$ antennas at the receiver and $N$ antennas at the source. This is  described by the following discrete-time input-output relation~\cite{TseBook}:
\begin{equation}\label{eq:MIMO_channel}
{\bf y} = {\bf H}{\bf x} + {\bf n}
\end{equation} 
where ${\bf y}\in \mathbb{C}^{M}$ and ${\bf x}\in \mathbb{C}^{N}$ denote the output and input vectors, respectively. The vector ${\bf x}$ must satisfy $\mathbb{E}\{{{\bf x}^{\Htran}{\bf x}}\}\le P_{\rm T}$ to constrain the total transmit power. Also, ${\bf n}\sim \mathcal{N}_\mathbb{C}({\bf 0}, {\bf R}_n)$ is the additive Gaussian noise and ${\bf H}\in \mathbb{C}^{M \times N}$ is the MIMO channel matrix. {The input-output relation in \eqref{eq:MIMO_channel} can be used to model a great variety of multiple antenna communication systems. In order to successfully model a particular one, there is the need to encode the physical context of the system into it. This is exactly the point where the circuit theoretic concept of linear multiports from~\cite{Nossek2014} comes into play.}
%
The physical model, based on the circuit theoretic approach, is shown in Fig.~\ref{fig: Network Model}. It consists of four basic parts: \emph{signal generation}, \emph{impedance matching}, \emph{antenna mutual coupling}, and \emph{noise}. The meaning of each part is briefly reviewed next. More details can be found in~\cite{Nossek2014}. 

\subsection{Signal generation and power} 
The generation of the $n$th physical signal that is to be transmitted is modeled by a voltage source, with complex envelope $v_{{\rm G},n}$, in series with the impedance \textcolor{blue}{$Z_{{\rm G}}=R_{{\rm G}}+\imagunit X_{{\rm G}}$}. The average available power of the voltage generator is $P_{{\rm a},n}=\frac{\Exp\{ |{v}_{{\rm G},n}|^{2}\}}{4R_{{\rm G}}}$
where the expectation accounts for signal randomness. Letting ${\bf v}_{{\rm G}} = [{v}_{{\rm G},1}, {v}_{{\rm G},2}, \ldots,{v}_{{\rm G},N}]^{T}$, the total average available power is thus
\begin{equation}\label{available_power}
P_{{\rm a}}=\sum\limits_{n=1}^NP_{{\rm a},n} = \frac{\Exp\{ {\bf v}^{\Htran}_{{\rm G}}{\bf v}_{{\rm G}}\}}{4R_{{\rm G}}}.
\end{equation}

\subsection{Impedance matrices} The transmit/receive matching networks are multiport systems described by the impedance matrices {\ZMT} and {\ZMR}. In particular, ${\bf Z}_{\rm MT}\in \mathbb{C}^{2N \times 2N}$ and ${\bf Z}_{\rm MR} \in \mathbb{C}^{2M \times 2M}$ are given by
\begin{equation}
\label{eq:ZMT}
\begin{split}
{\bf Z}_{\rm MT} = \left[\begin{array}{cc} {\bf Z}_{{\rm MT},11} & {\bf Z}_{{\rm MT},12} \\ {\bf Z}_{{\rm MT},21} & {\bf Z}_{{\rm MT},22} \end{array}\right] \\  {\bf Z}_{\rm MR} = \left[\begin{array}{cc} {\bf Z}_{{\rm MR},11} & {\bf Z}_{{\rm MR},12} \\ {\bf Z}_{{\rm MR},21} & {\bf Z}_{{\rm MR},22} \end{array}\right]
\end{split}
\end{equation}
with $\{{\bf Z}_{{\rm MT},ij}\in \mathbb{C}^{N \times N}; i =1,2,  j =1,2\}$ and $\{{\bf Z}_{{\rm MR},ij}\in \mathbb{C}^{M \times M}; i =1,2,  j =1,2\}$. We assume that the impedance matching networks are \textit{lossless}, {reciprocal}~\cite{Nossek2010}, and noiseless~\cite{Nyquist1928}. 

The impedance matrix ${\bf Z}_{\rm A} \in \mathbb{C}^{(N+M) \times (N+M)}$ accounts for the mutual coupling between antennas and can be partitioned as:
\begin{equation}
\label{eq:ZA}
{\bf Z}_{\rm A}=\left[\begin{array}{cc} {\bf Z}_{{\rm AT}} & {\bf Z}_{{\rm ATR}} \\ {\bf Z}_{{\rm ART}} & {\bf Z}_{{\rm AR}} \end{array}\right].
\end{equation}
Particularly, ${\bf Z}_{{\rm AT}} \in \mathbb C^{N \times N}$ and ${\bf Z}_{{\rm AR}} \in \mathbb{C}^{M \times M}$ quantify the mutual coupling at the transmit and receive sides (\textit{intra-array coupling}), respectively, while ${\bf Z}_{{\rm ATR}} \in  \mathbb{C}^{N \times M}$ and ${\bf Z}_{{\rm ART}} \in  \mathbb{C}^{M \times N}$ model the mutual coupling between the transmit  and the receive arrays (\textit{inter-array coupling}).  Because antennas are reciprocal, e.g.~\cite{balanis}, we have  ${\bf Z}_{{\rm AT}}={\bf Z}^{T}_{{\rm AT}}$, ${\bf Z}_{{\rm AR}}={\bf Z}^{T}_{{\rm AR}}$, and ${\bf Z}_{{\rm ATR}}={\bf Z}^{T}_{{\rm ART}}$.

A common approximation for ${\bf Z}_{\rm A}$ follows from the unilateral assumption, according to which ${\bf Z}_{{\rm ATR}} \approx {\bf 0}_{N \times M}$.
This basically implies that the currents at the receiver do not produce effects on the transmitter. From a mathematical standpoint, it requires that $||{\bf Z}_{{\rm AT}} {\bf i}_{\rm AT}|| \gg ||{\bf Z}_{{\rm ATR}} {\bf i}_{\rm AR}||$
where ${\bf i}_{\rm AT}$ and ${\bf i}_{\rm AR}$ are the vectors of currents at the transmitting and receiving arrays, respectively (see Fig.~\ref{fig: Network Model}). In practice, it implies that the transmitting antennas are not affected by the presence of the receiving antennas. This is true as long as the transmit and receive arrays are sufficiently separated in space as it happens in any practical communication network.\footnote{It may not hold true if different short-range applications are considered, e.g., near-field communications or short-range simultaneous wireless information and power transfer systems.}

\subsection{Losses in the antennas}
Notice that even though antennas may possess minimal loss, this can become significant when substantial electric currents are required to transmit a specific power. Particularly when dealing with close antenna spacing, the internal losses within the antenna can significantly impact the performance. A common way to account for this, it is to include a dissipation resistance that is connected in series. This implies that the impedance matrices ${ \bf Z}_{\rm AT}$ and ${\bf Z}_{\rm AR}$ must be replaced with
\begin{align}
{ \bf Z}_{\rm AT} &\to { \bf Z}_{\rm AT}+ R_{\rm d} {\bf I}_N \quad \quad
{ \bf Z}_{\rm AR} \to { \bf Z}_{\rm AR}+ R_{\rm d} {\bf I}_M.
\end{align}
If different dissipation resistances are used at the different antennas, then the matrices $R_{\rm d}{\bf I}_N$ and $R_{\rm d} {\bf I}_M$ should be replaced with diagonal matrices. For a half-wavelength dipole the expression of the dissipation resistance can be found in~\cite[Example 2.13]{balanis}.

\subsection{Noise sources}
The vector ${\bf v}_{{\rm EN}}$ accounts for the \textit{extrinsic noise} originating from the background radiation, and its entries represent the complex envelopes of the voltages that appear at the antenna ports when no currents flow, i.e., open-circuit noise voltages~\cite{Nossek2010}. The elements of ${\bf v}_{{\rm EN}}$ are zero-mean correlated random variables, with ${\bf R}_{\rm EN}=\Exp{\{{\bf v}_{{\rm EN}} {\bf v}^{\Htran}_{{\rm EN}}\}}$. A commonly adopted model is \cite[Sec. II-E]{Nossek2010}
\begin{equation}
\label{RA}
{\bf R}_{\rm EN}=4 k_{B} T_{\rm A} \Delta f \re({\bf Z}_{{\rm AR}})
\end{equation} 
where $k_{B}$ is the Boltzmann constant, $T_{\rm A}$ is the noise temperature of the antennas, while $\Delta f$ is the equivalent noise bandwidth that depends on the bandwidth of the desired signal.

The \textit{intrinsic noise} is produced by the subsystems that follow the receive matching network such as \textcolor{blue}{low noise amplifiers (LNAs), mixers, and analog-to-digital converters (ADCs)}. Most of the noise originates from the LNAs, and thus can be modelled by using the voltage and current vectors~\cite{Nossek2010,Rothe1956} given by ${\bf v}_{{\rm LNA}}$ and ${\bf i}_{{\rm LNA}}$, respectively. Both ${\bf v}_{{\rm LNA}}$ and ${\bf i}_{{\rm LNA}}$ are zero-mean random vectors, with the following statistics~\cite[Eq. (10)]{Nossek2010}: $\Exp{\{{\bf i}_{{\rm LNA}} {\bf i}^{\Htran}_{{\rm LNA}}\}}=\sigma^{2}_{i}{\bf I}_{M}$, $\Exp{\{{\bf v}_{{\rm LNA}} {\bf v}^{\Htran}_{{\rm LNA}}\}}=R^{2}_{\rm N} \sigma^{2}_{i} {\bf I}_{M}$ and 
\begin{equation}\label{CrossCorLna}
\Exp{\{{\bf v}_{{\rm LNA}} {\bf i}^{\Htran}_{{\rm LNA}}\}}=\rho R_{\rm N} \sigma^{2}_{i} {\bf I}_{M} \quad
\end{equation}
where $R_{\rm N}$ is the so-called \textit{noise resistance} of the LNAs, usually indicated in the manufacturer data sheets. The complex parameter 
\begin{equation} \label{ro}
\rho=\dfrac{\Exp\{{v}_{{\rm LNA},m}{i}^{\ast}_{{\rm LNA},m}\}}{\sqrt{\Exp\{|{v}_{{\rm LNA},m}|^{2}\}\Exp\{|{i}_{{\rm LNA},m}|^{2}\}}}
\end{equation}
accounts for the correlation between voltage and current noise generators at each port.


\subsection{Input-Output Relation}
Under the unilateral approximation, the input-output relation is~\cite[Eq. (16)]{Nossek2010}
\begin{equation}
\label{eq:ioTotal1}
{\bf v}_{\rm L}= {\bf D} {\bf v}_{\rm G} +{\boldsymbol \eta}
\end{equation}
where ${\bf D}$ and ${\boldsymbol \eta}$ are given by~\cite[Eq. (17)]{Nossek2010}
\begin{align}
\label{matD}
{\bf D}&={\bf Q}\,{\bf Z}_{{\rm RT}}(Z_{\rm G}{\bf I}_N+{\bf Z}_{{\rm T}})^{-1}\\
\label{veceta}
{\boldsymbol \eta}&= {\bf Q}({\bf F}_{\rm R}{\bf v}_{\rm EN}-{\bf v}_{\rm LNA}+{\bf Z}_{{\rm R}}{\bf i}_{\rm LNA})
\end{align}
with~\cite[Eq. (19)]{Nossek2010}
\begin{align}
\label{ZR}
{\bf Z}_{{\rm R}}&={\bf Z}_{{\rm MR},11}-{\bf F}_{{\rm R}}{\bf Z}_{{\rm MR},21}\\\label{ZT}
{\bf Z}_{{\rm T}}&={\bf Z}_{{\rm MT},11}-{\bf F}_{{\rm T}} {\bf Z}_{{\rm MT},21}\\
\label{ZRT}
{\bf Z}_{{\rm RT}}&= {\bf F}_{{\rm R}}  {\bf Z}_{{\rm ART}}  {\bf F}_{{\rm T}} ^{^{\Ttran}} \end{align}
and~\cite[Eq. (20)]{Nossek2010}
\begin{align}\label{matF_R}
{\bf F}_{{\rm R}} &= {\bf Z}_{{\rm MR},12}\left({\bf Z}_{{\rm MR},22} + {\bf Z}_{{\rm AR}}\right)^{-1}\\\label{matF_T}
{\bf F}_{{\rm T}} &={\bf Z}_{{\rm MT},12} \left({\bf Z}_{{\rm MT},22} + {\bf Z}_{{\rm AT}}\right)^{-1} \\\label{matQ}
{\bf Q}&=Z_{\rm L}(Z_{\rm L}{\bf I}_{M}+{\bf Z}_{{\rm R}})^{-1}.
\end{align}
The input-output relation~\eqref{eq:ioTotal1} can be written in a slightly different form (which will turn useful later on) as ${\bf v}_{\rm L} = {\bf Q}({\bf F}_{{\rm R}} {\bf v}_{\rm OC} +\tilde{\boldsymbol \eta})$
where $\tilde{\boldsymbol \eta}={\bf F}_{\rm R}{\bf v}_{\rm EN}-{\bf v}_{\rm LNA}+{\bf Z}_{{\rm R}}{\bf i}_{\rm LNA}$
and
\begin{equation}
\label{voc}
{\bf v}_{\rm OC}= {\bf D}_{\rm OC} {\bf v}_{\rm G}\mathop{=}^{(a)} {\bf Z}_{{\rm ART}} {\bf i}_{{\rm AT}}
\end{equation}
with \begin{equation}\label{doc}
{\bf D}_{\rm OC} =  {\bf Z}_{{\rm ART}}  {\bf F}_{{\rm T}}^{\Ttran} (Z_{\rm G}{\bf I}_N+{\bf Z}_{{\rm T}})^{-1}.
\end{equation}
Notice that ${\bf v}_{\rm OC}$
is the \textit{open circuit} voltage vector as induced by ${\bf i}_{{\rm AT}}$ when ${\bf i}_{{\rm AR}}={\bf 0}$, as it follows from $(a)$ in~\eqref{voc}. In general, ${\bf i}_{{\rm AR}}={\bf 0}$ does not imply that the elements of ${\bf v}_{\rm OC}$ are the same as if the receive antennas were isolated. \textcolor{blue}{This holds true only if the receive antennas are \textit{canonical minimum scattering} (CMS) antennas.}\footnote{According to~\cite{Kahn1965}, \textit{a canonical minimum-scattering antenna is “invisible” when the accessible waveguide terminals are open-circuited.} This means that a vanishing electric current in the antenna does not alter the electromagnetic field.} This is the case of half-wavelength dipoles~\cite{Laas2020}. We finally notice that ${\bf D} {\bf v}_{\rm G} = {\bf Q}{\bf F}_{{\rm R}} {\bf v}_{\rm OC}$ so that, using~\eqref{voc}, we get
 \begin{equation}\label{eq:D_OC_relation}
{\bf D} =  {\bf Q}{\bf F}_{{\rm R}}{\bf D}_{\rm OC}.
\end{equation}

\begin{remark}(Input-output relation without matching networks)
In the absence of a transmit matching network, the input-output relation can simply be obtained by setting ${\bf Z}_{\rm T}={\bf Z}_{\rm AT}$ and ${\bf F}_{\rm T}={\bf I}_{N}$ in~\eqref{matD} and~\eqref{ZRT}, respectively. Analogously, with no receive matching network the input-output relation can be obtained by replacing ${\bf Z}_{\rm R}$ with ${\bf Z}_{\rm AR}$ and ${\bf F}_{\rm R}$ with ${\bf I}_{M}$.
\end{remark}

\subsection{Transmit power and noise covariance matrix} The \textit{transmit power} is defined as the average active power at the output of the transmit matching network, or equivalently, at the input of the transmit antenna array, i.e., \textcolor{blue}{$P_{\rm T}= \frac{1}{2}\Exp\{\re({\bf v}^{\Htran}_{\rm AT}{\bf i}_{\rm AT}) \}$}.
Assuming a lossless transmit matching network, we have $\Exp\{\re({\bf v}^{\Htran}_{\rm AT}{\bf i}_{\rm AT}) \} = \Exp\{\re({\bf v}^{\Htran}_{\rm T}{\bf i}_{\rm T}) \}$. Now, observe that, under the unilateral approximation, ${\bf v}_{\rm T} = {\bf Z}_{\rm T} {\bf i}_{\rm T}$. Since ${\bf v}_{\rm T} = {\bf v}_{\rm G} - Z_{\rm G} {\bf i}_{\rm T}$, we obtain ${\bf i}_{\rm T} = \left(Z_{\rm G} {\bf I}_{N} + {\bf Z}_{\rm T} \right)^{-1}{\bf v}_{\rm G}$
so that $P_{\rm T}$ reduces to
\textcolor{blue}{\begin{equation}
\label{eq:P_T_new}
P_{\rm T}=\frac{1}{2R_{\rm G}}\Exp\{\re({\bf v}_{\rm G}^{\Htran}{\bf B}{\bf v}_{\rm G}) \}
\end{equation}}
with 
\textcolor{blue}{\begin{equation}
\label{matB2}
{\bf B}=R_{\rm G} (Z_{\rm G}{\bf I}_{N}+{\bf Z}_{{\rm T}})^{-\Htran} \re\{{\bf Z}_{{\rm T}}\}  (Z_{\rm G}{\bf I}_{N}+{\bf Z}_{{\rm T}})^{-1}.
\end{equation}}
Notice that $P_{\rm T}$ coincides with the \textit{radiated power} $P_{\rm rad}$ only if the transmit antennas are lossless.

From the statistics of the extrinsic and intrinsic noise, the covariance matrix ${\bf R}_{\eta}$ of ${\boldsymbol \eta}$ in~\eqref{veceta} is:
\begin{equation}
\label{Reta}
{\bf R}_{\eta}={\bf Q}{\bf U}{\bf Q}^{\Htran}
\end{equation}
where ${\bf Q}$ is given in~\eqref{matQ} and ${\bf U}= {\bf U}_{\rm IN} + {\bf U}_{\rm EN}$ is the correlation matrix of $\tilde{\boldsymbol \eta}$ with
\begin{align}\label{Upsilon2_IN}
{\bf U}_{\rm IN}&= \sigma^{2}_{\rm i}\left({\bf Z}_{{\rm R}}{\bf Z}^{\ast}_{{\rm R}}  - 2 R_{\rm N} \re \left( \rho^{\ast} {\bf Z}_{{\rm R}} \right) + R_{\rm N}^2{\bf I}_M\right)
\end{align}
and ${\bf U}_{\rm EN}=  {\bf F}_{\rm R} {\bf R}_{\rm EN} {\bf F}^{\Htran}_{\rm R}$.

\subsection{Matching network optimization}
The transmit matching network ${\bf Z}_{{\rm MT}}$ can be designed to maximize the power delivered to antennas (\textit{power matching} or \emph{maximum power transfer})~\cite{Nossek2010}. This yields ${\bf B}={\bf I}_N$ in~\eqref{matB2}, and 
\begin{align}\label{Z_T_opt}
{\bf Z}_{\rm T} = Z^{\ast}_{\rm G}{\bf I}_N. 
\end{align}
By taking~\eqref{ZT} and~\eqref{matF_T} into account, this can be obtained by setting~\cite{Nossek2010}
\begin{align}\label{Z_MT}
\!\!\!{\bf Z}_{{\rm MT}}^\star= \begin{bmatrix}
-\mathrm{j}{X_{\rm G}}{\bf I}_{N} & -\mathrm{j} \sqrt{R_{\rm G}}\re\{{\bf Z}_{{\rm AT}}\}^{1/2}  \\
-\mathrm{j} \sqrt{R_{\rm G}}\re\{{\bf Z}_{{\rm AT}}\}^{1/2} & -\mathrm{j}{\imag\{{\bf Z}_{\rm AT}\}}  \\
\end{bmatrix}
\end{align}
which yields
\begin{align}\label{F_T_matched}
{\bf F}_{{\rm T}}= -\mathrm{j} \sqrt{R_{\rm G}}\re\{{\bf Z}_{{\rm AT}}\}^{-1/2}.
\end{align}
The receive matching network ${\bf Z}_{{\rm MR}}$  can be  designed to  ensure that the signal-to-noise ratio (SNR) is as large as it can be (\textit{noise matching} or \emph{SNR maximization})~\cite{Nossek2010}. This is achieved with
\begin{align}\label{Z_MR}
{\bf Z}_{{\rm MR}}^\star&= \begin{bmatrix}
\mathrm{j}{\imag\{Z_{\rm opt}\}}{\bf I}_M & \mathrm{j} \sqrt{\re\{Z_{\rm opt}\}}\re\{{\bf Z}_{{\rm AR}}\}^{1/2}  \\
\mathrm{j} \sqrt{\re\{Z_{\rm opt}\}}\re\{{\bf Z}_{{\rm AR}}\}^{1/2} & -\mathrm{j}{\imag\{{\bf Z}_{\rm AR}\}}  \\
\end{bmatrix}.\end{align}
Plugging~\eqref{Z_MR} into~\eqref{ZR} and~\eqref{matF_R} yields
\begin{align}
\label{ZR_matched}
{\bf Z}_{{\rm R}} = Z_{\rm opt} {\bf I}_M\end{align}
with $Z_{\rm opt} = R_{\rm N} \left(\sqrt{1 - (\imag\{\rho\})^2} + \mathrm{j}\imag\{\rho\}\right)$ and 
\begin{align}\label{FR_matched}
{\bf F}_{{\rm R}} = \mathrm{j} \sqrt{\re\{Z_{\rm opt}\}}\re\{{\bf Z}_{{\rm AR}}\}^{-1/2}.
\end{align}
Also, notice that 
\begin{align}\label{Q_matched}
{\bf Q} = \frac{Z_{\rm L}}{Z_{\rm L}+Z_{\rm opt}} {\bf I}_M 
\end{align}
and the covariance matrix ${\bf R}_{\eta} = |Z_{\rm L}|^2 |Z_{\rm L}+Z_{\rm opt}|^{-2}\sigma^2 {\bf I}_M$
becomes diagonal with 
\begin{align}\label{sigma_2}
\begin{split}
\sigma^2&=\sigma^{2}_{i}\left(|Z_{\rm opt}|^2 - 2 R_{\rm N} \re \left( \rho^{\ast} Z_{\rm opt}\right)+R_{\rm N}^2\right) \\ &+ 4 k_{B} T_{A} \Delta f \re\{Z_{\rm opt}\}).
\end{split}
\end{align}
The design of coupled matching networks is very challenging for arrays with a large number of antennas \cite{Nossek2010,Nossek2014}. A practical approach is to make use of a \emph{self-impedance matching network}~\cite[Sect. III.B]{Warnick2009}, instead of a full multiport matching network. This approach neglects the mutual coupling among antennas and replaces, in the design of the matching networks, the impedance matrices ${\bf Z}_{{\rm AT}}$ and ${\bf Z}_{{\rm AR}}$ with the diagonal matrices $\diag({\bf Z}_{{\rm AT}})$ and $\diag({\bf Z}_{{\rm AR}})$ that contain only their diagonal elements. It becomes possible to substitute these matrices for the actual ones in~\eqref{Z_MT} and~\eqref{Z_MR} and specify uncoupled matching networks, as described above. 

\begin{table*}[t!]
\renewcommand{\arraystretch}{1.2}
\centering
\caption{Parameters of the antenna array at the BS.}
\begin{tabular}{c|c|c|c}
{ \bf Parameter} & {\bf Value}&{ \bf Parameter} & {\bf Value}  \\
\hline
Carrier frequency & $3.5$ GHz &Variance of the current noise source & $\sigma^2_i = 2 k_B B_W T_{\rm A}/R_{\mathrm N}$\\
\hline
Bandwidth  & $B_W = 20$ MHz & Radiation resistance   & $R_{\mathrm r} = 73$ \si{\ohm} \\
\hline
Transmit power & $P_T=-30$ dBW &LNA noise resistance & $R_{\mathrm N} = 5\,\si{\ohm}$  \\
\hline
Amplifier and load impedance  & $Z_{\rm G}= Z_{\rm L}= (186 - {\mathrm j}31.6)$ \si{\ohm} &Complex correlation coefficient& $\rho =0.1$ \\
\hline
Noise temperature of antennas  & $T_{\rm A}= 290$ &Dissipation resistance   & $R_{\mathrm d} = 10^{-3}R_{\mathrm r}$ \si{\ohm} \\
\hline
\end{tabular}
\label{tab:array_parameters}
\end{table*}

\section{Computation of the Mutual Coupling Impedance Matrix}
Next, we show how to compute the mutual coupling impedance matrix ${\bf Z}_{\rm A}$. In particular, we assume that the antennas at both sides are half-wavelength dipoles of length $l_{d}=\lambda/2$ and radius $a_{d} \ll l_{d}$. \textcolor{blue}{Specifically, we set $a_{d} = 10^{-4}l_{d}$.} Moreover, we assume that the receiver is equipped with a uniform linear array.

\vspace{-0.3cm}
\subsection{Impedance matrix ${\bf Z}_{{\rm AR}}$} 
We consider ${\bf Z}_{\rm AR}$ but the same analysis follows for ${\bf Z}_{\rm AT}$.
The mutual impedance between dipole $p$ and dipole $q$ is computed as \cite[Eq. (25.4.14)]{orfanidis}
\begin{equation}
\label{Zat,mn}
[{\bf Z}_{\rm AR}]_{pq}=-\dfrac{1}{I_{p}I_{q}}\int\limits_{-l_{d}/2}^{l_{d}/2} e_{qp}(s) I_{q}(s) ds
\end{equation}
where $e_{qp}(s)$ is the component (along the direction of dipole $q$) of the electric field produced by a current $I_{p}(s')$ flowing in dipole $p$, $I_{q}(s)$ is the current flowing in dipole $q$, and finally $I_{p}$ and $I_{q}$ are the currents at the input terminals of dipoles $p$ and $q$, respectively. The current distributions $I_{p}(s')$ and $I_{q}(s)$, for $p,q=1,2,\ldots,N$,  can be found by solving a system of Hall\'en integral equations~\cite[Sec.~25.7]{orfanidis}. 
The solutions can be found by numerical methods (e.g. the method of moments discussed in~\cite[Sec.~24.8]{orfanidis}). A very good approximation of the current distributions $I_{n}(s)$ for \textit{center-fed} dipoles is represented by the \textit{sinusoidal current model}, i.e.,
\begin{equation}
\label{Insin}
I_{n}(s)=I_{n} \dfrac{\sin \left[k\left({l_{d}}/{2}-|s|\right)\right]}{ \sin(kl_{d}/2)}.
\end{equation}
where $k=2 \pi/\lambda$ is the wavenumber. Based on~\eqref{Insin}, closed form expressions for the mutual impedance can be found in~\cite[Eqs. (8.69) and (8.71a-b)]{balanis} for dipoles in \textit{side-by-side configuration}. Closed form expressions are also provided for dipoles in \textit{collinear configuration}~\cite[Eqs. (8.72a-b)]{balanis}, and in \textit{parallel-in-echelon configuration}~\cite[Eqs. (8.73a-b)]{balanis}. As for the self impedance, which coincides with the diagonal elements of ${\bf Z}_{{\rm AR}}$, this can be found in~\cite[Eqs. (8.60a-b) and (8.61a-b)]{balanis}. 

\begin{figure}
\centering
\begin{subfigure}{0.48\textwidth}
\includegraphics[width = \columnwidth]{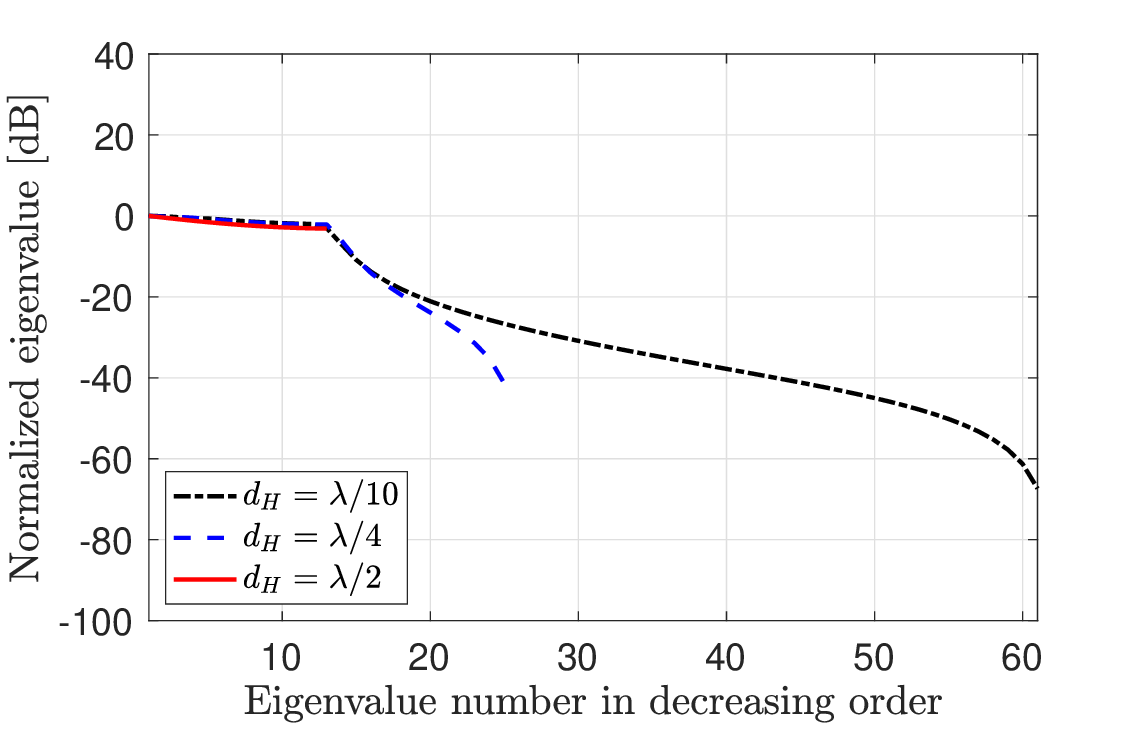}
\caption{Impedance matrix ${\bf Z}_{\rm AR}$}
\label{fig:Eigenvalues_ZAR}
\end{subfigure}
\begin{subfigure}{.48\textwidth}
\centering
\begin{overpic}[width = \columnwidth]{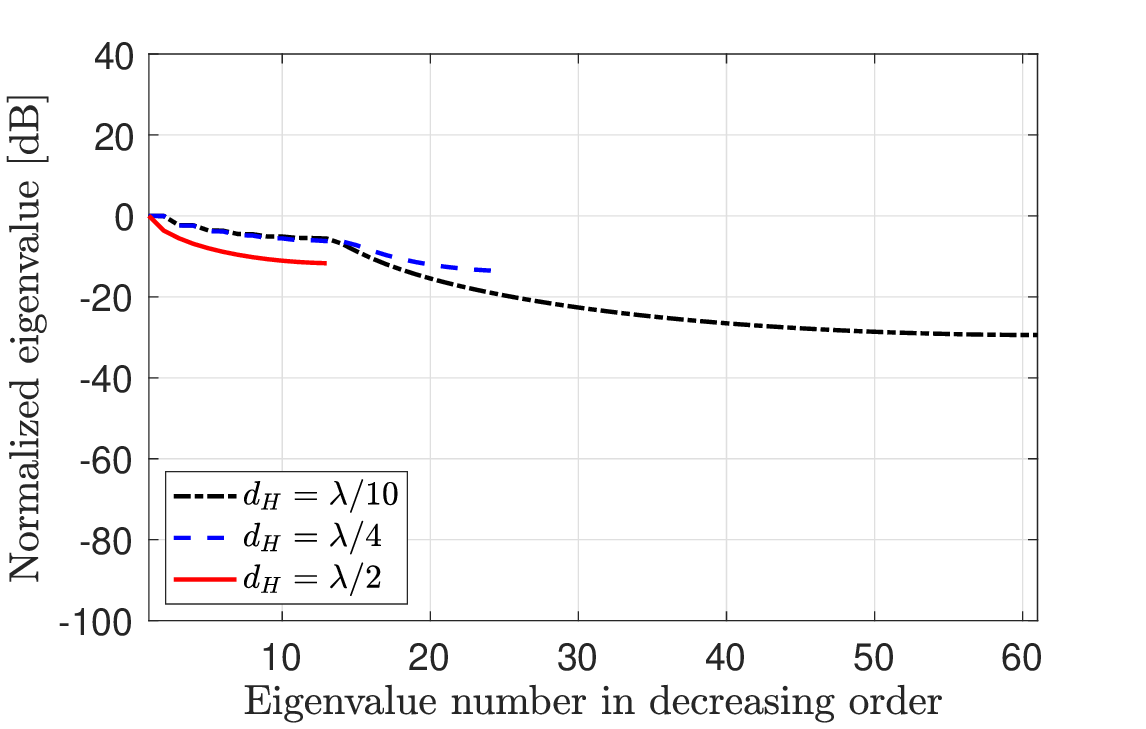}
\end{overpic}
\caption{Noise covariance matrix ${\bf U}$ without matching networks.}
\label{fig:Eigenvalues_U}
\end{subfigure}
\caption{Normalized eigenvalue distribution of ${\bf Z}_{\rm AR}$ and ${\bf U}$.}
\label{fig:NMSE}
\end{figure}

Fig.~\ref{fig:Eigenvalues_ZAR} shows the normalized eigenvalues of ${\bf Z}_{{\rm AR}}$ for an array of $\lambda/2$-dipoles in side-by-side configuration, for three different values of the inter-element spacing, namely $d_H = \lambda/10, \lambda/4$ and $\lambda/2$. The array size is $L_H = 6 \lambda$, and the number of array elements (and hence of eigenvalues) is $L_H / d_H +1$. Matrix ${\bf Z}_{{\rm AR}}$ has been calculated by using the sinusoidal model \eqref{Insin}. 
We see that the number of significant eigenvalues does not change appreciably when $d_H$ decreases below $\lambda/2$, and is approximately $2L_H /\lambda +1$. This means that for $d_H < \lambda/2$ mutual coupling introduces a significant correlation between the different array elements, as expected. 

Fig.~\ref{fig:Eigenvalues_U} shows the normalized eigenvalues of $\bf U$ in \eqref{Reta}, obtained without a matching network (see Remark 1) and with the parameter values reported in Table~\ref{tab:array_parameters}, \textcolor{blue}{e.g., \cite{Nossek2010} and references therein.} The behavior of these eigenvalues is quite different from that of Fig.~\ref{fig:Eigenvalues_ZAR}, because noise correlation not only depends on ${\bf Z}_{{\rm AR}}$ but also on the LNA parameters, and on the presence (and type) of matching networks, as shown in \eqref{Upsilon2_IN}. In particular, the curve corresponding to $d_H=\lambda/2$ seems to indicate that a significant correlation exists between the elements of $\tilde{\boldsymbol{\eta}}$ even with a half-wavelength spacing between the antennas.    

\begin{remark}
Other models for ${\bf Z}_{{\rm AR}}$ rely on the assumption of isotropic antennas or Hertzian dipoles~\cite{yordanov2009arrays,Ivrlac2009ICC,Nossek2010,MullerISIT2012,Nossek2014,Friedlander2020}. Isotropic antennas are inconsistent with the Maxwell equations \cite[Sec.~3.2]{kraus1988antennas}. \textcolor{blue}{Hence, models based on this assumption have no physical meaning \cite[Sec.~III]{Nossek2010}, though they are commonly used in the literature for analytical tractability.} 
In the case of Hertzian dipoles, a uniform current distribution is typically assumed, which approximates well the current distribution of an \textit{infinitesimal} linear wire $(l_{d} \le \lambda/50)$ with plates at its endpoints~\cite[Sec.~4.2]{balanis}. The sinusoidal model is a more accurate representation of the current distribution of any wire antenna~\cite[Sec.~4.3]{balanis}.\vspace{-0.3cm} 
\end{remark}

\vspace{-0.3cm}
\subsection{Impedance matrix ${\bf Z}_{{\rm ART}}$}

\textcolor{blue}{The impedance matrix ${\bf Z}_{{\rm ART}}$ accounts for the mutual coupling between the transmit and receive antennas. It represents the physical wireless propagation channel and can in principle be computed starting from any (e.g., deterministic or stochastic) channel model. Operationally, ${\bf Z}_{{\rm ART}}$ determines the \textit{open-circuit} voltage array response represented by ${\bf v}_{\rm OC}={\bf Z}_{{\rm ART}} {\bf i}_{\rm AT}$, which is obviously influenced by various factors, including the type of antennas, array configuration, polarization and transmission medium. Notice that ${\bf v}_{\rm OC}={\bf Z}_{{\rm ART}} {\bf i}_{\rm AT}={\bf D}_{\rm OC} {\bf v}_{\rm G}$. Hence, it can also be obtained from ${\bf D}_{\rm OC}$. 
In Appendix A, we consider an arbitrary array of CMS antennas located in the far-field region (Assumption 1) of a single transmit antenna, so that ${\bf Z}_{{\rm ART}}$ and ${\bf D}_{\rm OC}$ reduce to the vectors ${\bf Z}_{{\rm ART}}$ and ${\bf D}_{\rm OC}$, respectively. We also assume that the transmission takes place in a LoS propagation scenario, and that the electromagnetic wave generated by the transmitting antenna and incident on an antenna of the receiving array can be approximated locally (i.e., at each receiving element) by a plane wave (Assumption 2). Let $(\theta_m,\phi_m)$ denote the direction of arrival of the plane wave incident on the $m$th receive antenna, and let $r_m$ denote the distance between its center and that of the transmit antenna. Under the above conditions, in Appendix A we show that
\begin{equation}
\label{eq:d_OC}
{\bf d}_{{\rm OC}} = \boldsymbol{\alpha}(\boldsymbol{\psi},{\bf r})\odot{\bf a}({\bf r})
\end{equation}
where $\boldsymbol{\psi}$ and ${\bf r }$ are vectors collecting the directions of arrival and the distances, $\boldsymbol{\alpha}(\boldsymbol{\psi},{\bf r})$ is the vector of the channel gains, and ${\bf a}(\bf r)$ is the array response vector. Their expressions can be found in Appendix A. Notice that \eqref{eq:d_OC} is a quite general model for LoS propagation that applies to arbitrary array configurations and accounts for: \emph{i}) the distances to the different antennas over the array; \emph{ii}) the effective antenna length; \emph{iii}) the losses from polarization mismatch.}

\textcolor{blue}{If the transmitting antenna is in the far-field of the receiving array, the well-known \textit{planar wave approximation} can be achieved \cite{Friedlander2019}. In this scenario, the direction of arrival $(\theta,\phi)$ is aligned with the direction of the line connecting the centers of the array and the transmit antenna, and $r_m$ is replaced by $r$, representing the distance between the two centers. Consequently, \eqref{eq:d_OC} reduces to (e.g., \cite{Friedlander2019})
\begin{equation}
 \label{doc}
 {\bf d}_{{\rm OC}} =  \alpha(\theta,\phi,r) {\bf a}(\theta,\phi)
\end{equation}
where $\alpha(\theta,\phi,r)$ is defined in Appendix A and ${\bf a}(\theta,\phi) = [ e^{\imagunit {\bf k}(\theta,\phi) \cdot {\boldsymbol \delta}_1},\ldots, e^{\imagunit {\bf k}(\theta,\phi)\cdot {\boldsymbol \delta}_M}]^{\Ttran}$ where ${\boldsymbol \delta}_m$ is the displacement vector from the array center to the center of the $m$th receive antenna and ${\bf k}(\theta,\phi)$ is the wave vector~\cite[eq. (17.1.4)]{orfanidis}.}

\section{Holographic MIMO Communications} \label{SectionIV}
We consider a communication system where the BS is equipped with $M_{\rm BS}$ antennas and serves $K$ single-antenna UEs. The uplink and downlink signal models are derived on the basis of the multiport communication model provided in~\eqref{eq:ioTotal1}, by taking into account that in the uplink $N=1$ and $M=M_{\rm BS}$ while $N=M_{\rm BS}$ and $M = 1$ in the downlink. \textcolor{blue}{We assume that lossless matching networks are used at each UE in uplink (i.e., for \emph{power matching}) and downlink (i.e., for \emph{noise matching}). This is reasonable since a single antenna is used at each UE.}


\subsection{Uplink data transmission}

In the uplink, the vector ${\bf v}_{\rm L} \in \mathbb{C}^{M_{\rm BS}}$ of voltages measured at the BS is generated by the superposition of the generator's voltages $\{{v}_{{\rm G},i}; i=1,\ldots,K\}$ of the $K$ single-antenna (i.e., $N=1$) transmitting UEs. The \emph{dimensionless} input-output relation can be obtained from~\eqref{eq:ioTotal1} as
\begin{equation}
\label{eq:ioTotal1_ul}
\frac{{\bf v}_{\rm L}^{\rm ul}}{\sqrt{c}}=\sum_{i=1}^{K}  {\bf d}_i^{\rm ul} \frac{{v}_{{\rm G},i}^{\rm ul}}{\sqrt{c}} +\frac{{\boldsymbol \eta}^{\rm ul}}{\sqrt{c}}
\end{equation}
where $c$ is an arbitrary constant, measured in ${\rm V}^2$ \textcolor{blue}{(Volts$^2$)}, needed to obtain a dimensionless relationship. The vector ${\bf d}_i^{\rm ul}\in \mathbb{C}^{M_{\rm BS}}$ associated with the single-antenna UE $i$ is obtained from~\eqref{eq:D_OC_relation} and reads 
\begin{align}\label{d_i_ul}
\begin{split}
{\bf d}_{i}^{\rm ul} & = {\bf Q}^{\rm ul}{\bf F}_{{\rm R}}^{\rm ul}{\bf d}_{{\rm OC},i}^{\rm ul} \mathop{=}^{(a)} \frac{{ F}_{{\rm T}}^{\rm ul}} {Z_{\rm G}^{\rm ul}+{ Z}_{{\rm T}}^{\rm ul} }{\bf Q}^{\rm ul}{\bf F}_{{\rm R}}^{\rm ul}  {\bf z}_{{\rm ART},i}^{\rm ul}  \\ & \mathop{=}^{(b)} \alpha_{\rm ul}(Z_{\rm L}^{\rm ul}{\bf I}_{M_{\rm BS}}+{\bf Z}_{{\rm R}}^{\rm ul})^{-1}{\bf F}_{{\rm R}}^{\rm ul}  {\bf z}_{{\rm ART},i}^{\rm ul}
\end{split}
\end{align}
where $(a)$ follows from~\eqref{voc} whereas $(b)$ is because a matching network for maximum power transfer is used by UE $i$. From~\eqref{Z_T_opt} and~\eqref{F_T_matched}, this implies ${Z}_{\rm T}^{\rm ul} = {(Z^{{\rm ul}}_{\rm G})}^{\ast}$ and ${F}_{{\rm T}}^{\rm ul} = -\mathrm{j} \sqrt{R_{\rm G}}\re\{{Z}_{{\rm AT}}^{\rm ul} \}^{-1/2}$, where ${Z}_{{\rm AT}}^{\rm ul} $ is the transmitting antenna impedance. In~\eqref{d_i_ul}, we have defined
\begin{equation}
\label{alfa_ul}
\alpha_{\rm ul} = -\frac{\mathrm{j}Z_{\rm L}^{\rm ul}}{2\sqrt{{R_{\rm G}^{\rm ul}}{\re\{{Z}_{{\rm AT}}^{\rm ul}\}}}}.
\end{equation}
From~\eqref{available_power} and~\eqref{eq:P_T_new}, the transmit power of a single-antenna UE $i$ can be computed as $P_{{\rm T},i}= b^{\rm ul} P_{{\rm a},i}$
where \textcolor{blue}{$P_{{\rm a},i} = \frac{1}{4R_{\rm G}}\Exp\{|v_{{\rm G},i}|^{2} \}$} is its available power and
\begin{equation}
\label{b}
\textcolor{blue}{b^{\rm ul}=2R_{\rm G}\dfrac{\re(Z_{\rm T}^{\rm ul})}{|Z_{\rm G}^{\rm ul}+Z_{\rm T}^{\rm ul}|^{2}}\mathop{=}^{(a)} \dfrac{1}{2}}
\end{equation}
is the fraction of available power delivered to the transmitting antenna. Notice that $(a)$ follows because ${Z}_{\rm T}^{\rm ul} = {(Z^{{\rm ul}}_{\rm G})}^{\ast}$ when a matching network for maximum power transfer is used by UE $i$.

\begin{table*}[t!]\vspace{0cm}
\renewcommand{\arraystretch}{1.3}
\centering
\caption{Relationship between uplink and downlink channels with different matching designs.}\vspace{0cm}
\begin{tabular}{c|c|c|c}
{ \bf Channel} & {\bf Arbitrary matching networks} & {\bf Without matching networks } & {\bf With full matching networks } \\
\hline
${\bf d}_k^{\rm dl}$& $ \dfrac {\alpha_{\rm dl}} {\alpha_{\rm ul}} {\bf A}_{\rm dl,ul} {\bf d}_k^{\rm ul} $  & $ \dfrac {\alpha_{\rm dl}} {\alpha_{\rm ul}}  {\bf d}_k^{\rm ul}
$ & $\dfrac {\xi_{\rm dl}} {\xi_{\rm ul}} {\bf d}_k^{\rm ul}$\\
\hline
${\bf h}_k^{\rm dl}$ & $ \dfrac {\alpha_{\rm dl}} {\alpha_{\rm ul}} {\bf B}_{\rm dl}^{-\Ttran/2}{\bf A}_{\rm dl,ul}{\bf h}_k^{\rm ul}$ & $ \dfrac {\alpha_{\rm dl}} {\alpha_{\rm ul}} {\bf B}_{\rm dl}^{-\Ttran/2} {\bf h}_k^{\rm ul}$ & $\dfrac {\xi_{\rm dl}} {\xi_{\rm ul}} {\bf h}_k^{\rm ul}$\\
\hline
\end{tabular}\vspace{0cm}
\label{tab:duality}
\end{table*}

By setting $\vect{y}^{\rm ul}= {\bf v}_{\rm L}/\sqrt{c}$, $\vect{h}_{i}^{\rm ul} = {\bf d}^{\rm ul}_i$, $x_i^{\rm ul} = {v}_{{\rm G},i} / \sqrt{c}$ and $\vect{n}^{\rm ul}={\boldsymbol \eta}^{\rm ul}/{\sqrt{c}}$, the input-output relation of the multi-user MIMO system in the form~\eqref{eq:MIMO_channel} follows:
\begin{align} \label{eq:uplink-signal-model_ul}
\vect{y}^{\rm ul}=  \sum_{i=1}^{K} \vect{h}_{i}^{\rm ul} x_i^{\rm ul} + \vect{n}^{\rm ul}.
\end{align}
 The data signal $x_i^{\rm ul}$ from UE~$i$ is modelled as $x_i^{\rm ul}\sim \CN({0}, p_i)$ with 
\begin{equation}
\label{TxPow2}
\textcolor{blue}{P_{{\rm T},i}=P_{{\rm a},i} = \dfrac{c}{2R_{\rm G}}p_i}.
\end{equation}
The vector $\vect{n}^{\rm ul}\sim \CN(\vect{0}_{M_{\rm BS}}, \vect{R}_{n}^{\rm ul})$ is independent noise with covariance matrix $\vect{R}_{n}^{\rm ul}= c^{-1}{\bf R}_{\eta}$, where ${\bf R}_{\eta}$ is given by~\eqref{Reta}.
Since $c$ is an arbitrary constant, we assume $c=1\,{\rm V}^2$ without loss of generality.

To decode $x_k^{\rm ul}$, the vector $\mathbf{y}^{\rm ul}$ is processed with the combining vector
$\mathbf{u}_k\in\mathbb{C}^{M_{\rm BS}}$. By treating the interference as noise, the spectral efficiency (SE) for UE $k$ is
$\log_2\left(1+\gamma_k^{\rm ul}\right)$ where
\begin{align}\label{eq:sinr_ul}
  \gamma_k^{\rm ul} = \frac{p_k\left|\mathbf{u}_k^{\Htran} \mathbf{h}_k^{\rm ul}\right|^2}
  { \sum_{i\neq k}{p_i\left|\mathbf{u}_k^{\Htran} \mathbf{h}_i^{\rm ul}\right|^2} + \mathbf{u}_k^{\Htran}\mathbf{R}_n^{\rm ul}\mathbf{u}_k}
\end{align}
is the SINR. We consider both MR and MMSE combining \cite{massivemimobook}. MR has low computational complexity and maximizes the power of the desired signal, but neglects interference. MMSE has higher complexity but it maximizes the SINR in~\eqref{eq:sinr_ul}. In the first
case, $\mathbf{u}_k=\mathbf{h}_k^{\rm ul}/\left\|\mathbf{h}_k^{\rm ul}\right\|$, while in the second case $  \mathbf{u}_k = ( \sum_{i=1}^{K}{{p_i\mathbf{h}_i^{\rm ul} {(\mathbf{h}_i^{\rm ul})}^{\Htran} }} + \mathbf{R}_n^{\rm ul})^{-1} \mathbf{h}_k^{\rm ul}$.

\subsection{Downlink data transmission} 
In the downlink, the voltage ${v}_{{\rm L},k}^{\rm dl} \in \mathbb{C}$ measured at the single antenna of UE $k$ is generated by the voltage vector ${\bf v}_{\rm G} \in \mathbb{C}^{M_{\rm BS}}$ at the BS array. From~\eqref{eq:ioTotal1}, the \emph{dimensionless} input-output relation is
\begin{equation}
\label{eq:ioTotal1_dl}
\frac{{v}_{{\rm L},k}^{\rm dl}}{\sqrt{c}}=  ({\bf d}_k^{\rm dl})^{\Ttran} \frac{{\bf v}_{\rm G}}{\sqrt{c}} +\frac{ \eta_k}{\sqrt{c}}
\end{equation}
where $c$ is an arbitrary constant measured in ${\rm V}^2$. The vector ${\bf d}_k^{\rm dl}\in \mathbb{C}^{M_{\rm BS}}$ is obtained from~\eqref{eq:D_OC_relation}:
\begin{align}\label{d_i_dl}
\begin{split}
{\bf d}_{i}^{\rm dl} &= {Q}^{\rm dl}{F}_{{\rm R}}^{\rm dl}{\bf d}_{{\rm OC},i}^{\rm dl} \\ &\mathop{=}^{(a)} {Q}^{\rm dl}{F}_{{\rm R}}^{\rm dl}  (Z_{\rm G}^{\rm dl}{\bf I}_{M_{\rm BS}}+{\bf Z}_{{\rm T}}^{\rm dl})^{-1} {\bf F}_{{\rm T}}^{\rm dl} {\bf z}_{{\rm ART},i}^{\rm dl} \\ & \mathop{=}^{(b)} \alpha_{\rm dl} (Z_{\rm G}^{\rm dl}{\bf I}_{M_{\rm BS}}+{\bf Z}_{{\rm T}}^{\rm dl})^{-1} {\bf F}_{{\rm T}}^{\rm dl} {\bf z}_{{\rm ART},i}^{\rm dl}
\end{split}
\end{align}
since the receiving UE has a single antenna. In particular, $(a)$ derives from~\eqref{doc} as ${\bf D}_{\rm OC}$ is a $1 \times {M_{\rm BS}}$ matrix (i.e., a row vector) whose transpose is exactly $(Z_{\rm G}^{\rm dl}{\bf I}_{M_{\rm BS}}+{\bf Z}_{{\rm T}}^{\rm dl})^{-1} {\bf F}_{{\rm T}}^{\rm dl} {\bf z}_{{\rm ART},i}^{\rm dl}$, whereas $(b)$ follows from~\eqref{ZR_matched} and~\eqref{FR_matched}. 
Also, we have defined
\begin{equation}
\label{alfa_dl}
\alpha_{\rm dl} = \frac{ \mathrm{j} Z_{\rm L}^{\rm dl} \sqrt{{\re\{Z_{\rm opt}^{\rm dl}\}}}} {(Z_{\rm L}^{\rm dl}+Z_{\rm opt}^{\rm dl}) \sqrt{{\re\{{Z}_{{\rm AR}}^{\rm dl}\}}}}.
\end{equation}
By setting ${y}_k^{\rm dl}= {v}_{{\rm L},k}^{\rm dl}/\sqrt{c}$, ${n}_k^{\rm dl}={\eta_k^{\rm dl}}/{\sqrt{c}}$ and
\begin{align}\label{eq:channel-model_dl}
\vect{h}_{k}^{\rm dl} &= {\bf B}_{\rm dl}^{-\Ttran/2}{\bf d}_k^{\rm dl} \\ {\bf x}^{\rm dl} &= \frac{1}{\sqrt{c}}{\bf B}_{\rm dl}^{1/2}{\bf v}_{\rm G} \label{eq:data-model_dl}
\end{align}
the input-output relation follows in the form
\begin{align} \label{eq:uplink-signal-model_dl}
{y}_k^{\rm dl}=  {(\vect{h}_k^{\rm dl})}^{\Ttran} {\bf x}^{\rm dl} + {n}_k^{\rm dl}
\end{align}
with $P_{\rm T}=\frac{1}{4R_{\rm G}} \Exp\{\re({\bf v}_{\rm G}^{\Htran}{\bf B}_{\rm dl}{\bf v}_{\rm G})\}=\frac{c}{4R_{\rm G}}\Exp\{ ||{\bf x}^{\rm dl} ||^2\} $.
Since a noise matching network is used at each UE, we have that {${n}_k^{\rm dl}\sim \CN(0, c^{-1}\sigma_{\rm dl}^2)$} with $\sigma_{\rm dl}^2$ given by~\eqref{sigma_2}.
The vector ${\bf x}^{\rm dl}$ is obtained as 
\begin{align}
{\bf x}^{\rm dl} = \sum_{i=1}^{K}{\bf w}_i{\tilde x}_{i}^{\rm{dl}}
\end{align}
where $\tilde x_i^{\rm dl}\sim \CN({0}, p_i)$ is the information-bearing signal and ${\bf w}_i$ is the precoding vector associated with UE $i$ that satisfies $\Exp\{ ||{\bf w}_i||^2\} = 1$ so that $\Exp\{ ||{\bf x}^{\rm dl}||^2\} = \sum_{i=1}^K p_i$ and $P_{\rm T} = \frac{c}{4R_{\rm G}}\sum_{i=1}^K p_i$. By treating the interference as noise, the downlink SE for UE $k$ is
$\log_2\left(1+\gamma_k^{\rm dl}\right)$, where
\begin{align}\label{eq:sinr_dl}
  \gamma_k^{\rm dl} = \frac{p_k\left|{\mathbf{w}^{\Htran}_k \mathbf{h}_k^{\rm dl}}\right|^2}
  { \sum_{i\neq k}{p_i\left|{\mathbf{w}_i^{\Htran} \mathbf{h}_k^{\rm dl}}\right|^2} + \sigma_{\rm dl}^2}
\end{align}
is the SINR for $c=1\,{\rm V}^2$. We assume that ${\bf w}_k = {\overline{\bf w}_k}/{||\overline{\bf w}_k||}$
and consider both MR precoding with $\overline{\bf w}_k=\mathbf{h}_k^{\rm dl}$ and MMSE precoding with $  \overline{\bf w}_k = ( \sum_{i=1}^{K}{{p_i\mathbf{h}_i^{\rm dl} {(\mathbf{h}_i^{\rm dl})}^{\Htran} }} + \sigma_{\rm dl}^2 {\bf I}_{M_{\rm BS}})^{-1} \mathbf{h}_k^{\rm dl}$.


\subsection{Uplink and downlink duality}


The concept of uplink and downlink duality in wireless communications refers to the relationship between the uplink and downlink channels in a communication system, with the exception of a scaling factor. The duality principle states that the uplink channel vector is proportional to the transpose of the downlink channel vector, with a scaling factor that depends on various factors (e.g., antenna gains). The significance of uplink and downlink duality lies in its practical implications for system design and optimization \cite{massivemimobook}. By exploiting this duality, system parameters and algorithms can be jointly designed for both uplink and downlink transmissions, simplifying system complexity and improving overall performance \cite[Sec. 4]{massivemimobook}. For example, channel estimation and combining techniques developed for uplink can be applied to the downlink without modification, leading to significant savings in complexity. Next, we will discuss this duality in three different cases at the BS: $1$) when arbitrary matching networks (e.g., self-impedance matching networks) are used, $2$) when no matching network is employed, and $3$) when full power and noise matching networks are employed. It should be noted that in all these cases, ${\bf z}_{{\rm ART},k}^{\rm ul}={\bf z}_{{\rm ART},k}^{\rm dl}$ holds, which is a result of the reciprocity principle in electromagnetic propagation.    
\smallskip
\subsubsection{Arbitrary matching networks at the BS} In general, when arbitrary matching networks are used at the BS (e.g., self-impedance matching networks), we have that ${\bf F}_{{\rm T}}^{\rm dl} \ne {\bf F}_{{\rm R}}^{\rm ul}$ and ${\bf Z}_{{\rm T}}^{\rm dl} \ne {\bf Z}_{{\rm R}}^{\rm ul}$. From~\eqref{d_i_ul} and~\eqref{d_i_dl}, the physical channels ${\bf d}_{k}^{\rm ul}$ and ${\bf d}_{k}^{\rm dl}$ exhibit reciprocity up to a linear transformation. Particularly, ${\bf d}_k^{\rm dl}$ can be obtained from ${\bf d}_k^{\rm dl}$ as
\begin{equation}
\label{eq:d_uld_dl_arbitrary_matching}
{\bf d}_k^{\rm dl} = \dfrac {\alpha_{\rm dl}} {\alpha_{\rm ul}} {\bf A}_{\rm dl,ul} {\bf d}_k^{\rm ul}
\end{equation}
where we have defined ${\bf A}_{\rm dl,ul} = (Z_{\rm G}^{\rm dl}{\bf I}_{M_{\rm BS}}+{\bf Z}_{{\rm T}}^{\rm dl})^{-1} {\bf F}_{{\rm T}}^{\rm dl} ({\bf F}_{{\rm R}}^{\rm ul})^{-1}(Z_{\rm L}^{\rm ul}{\bf I}_{M_{\rm BS}}+{\bf Z}_{{\rm R}}^{\rm ul})$. 
A similar reciprocity condition is evident for ${\bf h}_k^{\rm ul}$ and ${\bf h}_k^{\rm dl}$ since $\vect{h}_{k}^{\rm ul} = {\bf d}^{\rm ul}_k$ and $\vect{h}_{k}^{\rm dl} = {\bf B}_{\rm dl}^{-\Ttran/2}{\bf d}_k^{\rm dl}$. In particular, from~\eqref{eq:d_uld_dl_arbitrary_matching}, ${\bf h}_k^{\rm dl}$ can be obtained from $\vect{h}_{k}^{\rm ul}$ only if premultiplied by the matrix ${\bf B}_{\rm dl}^{-\Ttran/2}{\bf A}_{\rm dl,ul}$. 
The above results for ${\bf d}_k^{\rm dl}$ and ${\bf h}_k^{\rm dl}$ are summarized in the first column of Table \ref{tab:duality}.

\smallskip
\subsubsection{No matching networks at the BS} 
In the absence of matching networks at the BS, we have that ${\bf F}_{{\rm T}}^{\rm dl} = {\bf F}_{{\rm R}}^{\rm ul} = {\bf I}_M$, ${\bf Z}_{{\rm T}}^{\rm dl} = {\bf Z}_{{\rm AT}}^{\rm dl}$, ${\bf Z}_{{\rm R}}^{\rm ul} = {\bf Z}_{{\rm AR}}^{\rm ul}$. Accordingly,~\eqref{d_i_ul} and~\eqref{d_i_dl} become
\begin{align}\label{d_i_ul_woMR}
{\bf d}_{k}^{\rm ul} = \alpha_{\rm ul} (Z_{\rm L}^{\rm ul}{\bf I}_{M_{\rm BS}}+{\bf Z}_{{\rm AR}}^{\rm ul})^{-1} {\bf z}_{{\rm ART},k}^{\rm ul} \\\label{d_i_dl_woMR}
{\bf d}_{k}^{\rm dl} = \alpha_{\rm dl} (Z_{\rm G}^{\rm dl}{\bf I}_{M_{\rm BS}}+{\bf Z}_{{\rm AT}}^{\rm dl})^{-1} {\bf z}_{{\rm ART},k}^{\rm dl}.\end{align}
We observe that ${\bf z}_{{\rm ART},k}^{\rm ul}={\bf z}_{{\rm ART},k}^{\rm dl}$ and ${\bf Z}^{\rm dl}_{\rm AT} = {\bf Z}^{\rm ul}_{{\rm AR}}$ in~\eqref{d_i_ul_woMR} and~\eqref{d_i_dl_woMR}. Therefore, ${\bf d}_{k}^{\rm ul}={\bf d}_{k}^{\rm dl}$ when $Z_{\rm L}^{\rm ul} = Z_{\rm G}^{\rm dl}$. This condition is satisfied as it involves the load and generator impedances at the BS. 

As for ${\bf h}_k^{\rm ul}$ and ${\bf h}_k^{\rm dl}$, from~\eqref{matB2} we observe that, in the absence of a power matching network, ${\bf B}_{\rm dl}$ is no longer equal to the identity matrix  ${\bf I}_{M_{\rm BS}}$ but is given by
\begin{equation}\label{eq:B_dl}
{\bf B}_{\rm dl} =  R^{\rm dl}_{\rm G}(Z^{\rm dl}_{\rm G}{\bf I}_{M_{\rm BS}}+{\bf Z}^{\rm dl}_{{\rm AT}})^{-\Htran} \re\{{\bf Z}^{\rm dl}_{{\rm AT}}\}  (Z^{\rm dl}_{\rm G}{\bf I}_{M_{\rm BS}}+{\bf Z}^{\rm dl}_{{\rm AT}})^{-1}. 
\end{equation}
Hence, from~\eqref{d_i_ul_woMR} and~\eqref{d_i_dl_woMR} it follows that 
\begin{equation}
\label{eq:no_matching_duality}
{\bf h}_k^{\rm dl} = \dfrac {\alpha_{\rm dl}} {\alpha_{\rm ul}} {\bf B}_{\rm dl}^{-\Ttran/2} {\bf h}_k^{\rm ul}
\end{equation}
where we have used $\vect{h}_{k}^{\rm ul} = {\bf d}^{\rm ul}_k$, $\vect{h}_{k}^{\rm dl} = {\bf B}_{\rm dl}^{-\Ttran/2}{\bf d}_k^{\rm dl}$ and $Z_{\rm L}^{\rm ul} = Z_{\rm G}^{\rm dl}$. The equation above demonstrates that ${\bf h}_k^{\rm dl}$ can be derived from ${\bf h}_k^{\rm ul}$ by multiplying it with the matrix ${\alpha_{\rm dl}} /{\alpha_{\rm ul}}{\bf B}_{\rm dl}^{-\Ttran/2}$. The results for the downlink channel are summarized in the second column of Table \ref{tab:duality}.

\smallskip
\subsubsection{With power and noise matching networks at the BS}
When a noise matching network is used at the BS, the impedance matrix ${\bf Z}_{{\rm MR}}$ is equal to ${\bf Z}_{{\rm MR}}^\star$ in~\eqref{Z_MR}. From~\eqref{FR_matched} and~\eqref{Q_matched}, ${\bf d}_{k}^{\rm ul}$ in~\eqref{d_i_ul} becomes
\begin{equation}
\label{d_i_ul_MR}
{\bf d}_{k}^{\rm ul} = \xi_{\rm ul} \re\{{\bf Z}^{\rm ul}_{{\rm AR}}\}^{-1/2} {\bf z}_{{\rm ART},k}^{\rm ul}
\end{equation}
with
\begin{equation}
\xi_{\rm ul} = \dfrac{1}{2 \sqrt{R_{\rm G}^{\rm ul} \re\{ Z_{{\rm AT}}^{\rm ul} \}}} \dfrac{Z_{\rm L}^{\rm ul} \sqrt{\re\{Z_{\rm opt}^{\rm ul}\}} }{Z_{\rm L}^{\rm ul} + Z_{\rm opt}^{\rm ul}}.
\end{equation}
In the downlink, if a power matching network is used by the BS, then ${\bf Z}_{{\rm MT}}={\bf Z}_{{\rm MT}}^\star$ so that ${\bf Z}_{{\rm T}}^{\rm dl}$ and ${\bf F}_{{\rm T}}^{\rm dl}$ reduce to~\eqref{Z_T_opt}  and~\eqref{F_T_matched}. Hence,~\eqref{d_i_dl} becomes
\begin{align}\label{d_i_dl_MR}
{\bf d}_k^{\rm dl} = \xi_{\rm dl}  \re\{{\bf Z}^{\rm dl}_{{\rm AT}}\}^{-1/2} {\bf z}_{{\rm ART},k}^{\rm dl}
\end{align}
with 
\begin{equation}
\xi_{\rm dl} = \dfrac{1}{2 \sqrt{R_{\rm G}^{\rm dl}\re\{ Z_{{\rm AR}}^{\rm dl} \}}} \dfrac{Z_{\rm L}^{\rm dl} \sqrt{\re\{Z_{\rm opt}^{\rm dl}\}} }{Z_{\rm L}^{\rm dl} + Z_{\rm opt}^{\rm dl}}.
\end{equation}
Notice that ${\bf z}_{{\rm ART},k}^{\rm ul}={\bf z}_{{\rm ART},k}^{\rm dl}$. If the BS uses the same array for transmission and reception, then ${\bf Z}^{\rm dl}_{\rm AT} = {\bf Z}^{\rm ul}_{{\rm AR}}$. Putting together the above results yields
\begin{equation}
\begin{split}
\label{ }
\xi_{\rm dl}^{-1}{\bf d}_k^{\rm dl} &=\re\{{\bf Z}^{\rm dl}_{{\rm AT}}\}^{-1/2} {\bf z}_{{\rm ART},k}^{\rm dl} \\ &= \re\{{\bf Z}^{\rm ul}_{{\rm AR}}\}^{-1/2} {\bf z}_{{\rm ART},k}^{\rm ul} = \xi_{\rm ul}^{-1} {\bf d}_k^{\rm ul}
\end{split}
\end{equation}
which shows that ${\bf d}_k^{\rm ul}$ and ${\bf d}_k^{\rm dl}$ differ only for a scaling factor. This holds also for $\vect{h}_{i}^{\rm ul}$ and $\vect{h}_{i}^{\rm dl}$ since, in the presence of a power matching network, ${\bf B}_{\rm dl}$ reduces to ${\bf I}_{M_{\rm BS}}$ as it follows from~\eqref{matB2}. The results are summarized in the third column of Table \ref{tab:duality}. \textcolor{red}{The same result is obtained in the absence of a noise matching network at the BS by noting that ${\bf B}_{\rm dl}$ in \eqref{eq:no_matching_duality} depends only on ${\bf Z}^{\rm dl}_{{\rm AT}}$, as follows from \eqref{eq:B_dl}. In other words, reciprocity up to a scaling factor is achieved even if only the power matching network is used at the BS.}

\begin{figure}[t!]
\centering
\includegraphics[width = 1\columnwidth]{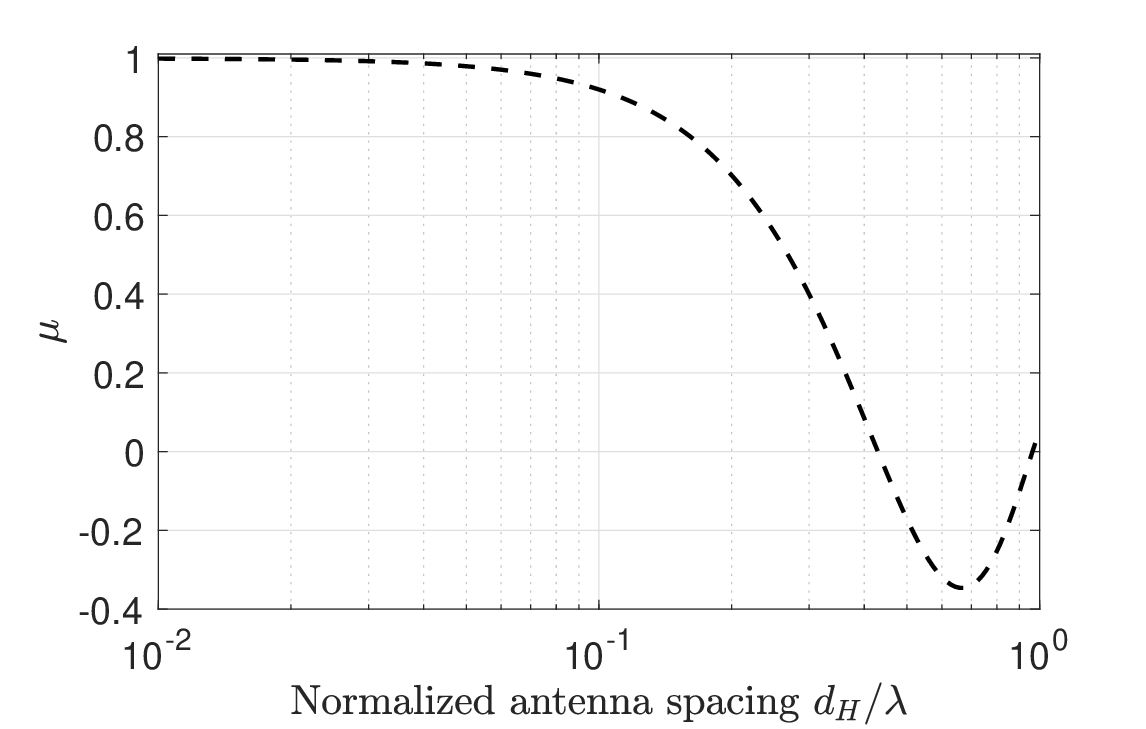}
\caption{ Behaviour of the normalized mutual coupling coefficient $\mu$ for two half-wavelength dipoles with sinusoidal current and side-by-side configuration as the spacing $d_H$ varies. The radiation resistance is $R_{\rm r} = 73$\si{\ohm} while the dissipation resistance is $R_{\rm d} = 10^{-3}R_{\rm r}$.
\vspace{-0.4cm}}
\label{fig:mu}
\end{figure}

\vspace{0 cm}
\section{The Effect of Coupling: A case study with two antennas in a single path LoS scenario}
To showcase what is the impact of mutual coupling in multi-user MIMO, next we consider a simple scenario in uplink with $K=2$ UEs and ${M_{\rm BS}} =2 $ half-wavelength dipoles in side-by-side
configuration. \textcolor{blue}{We consider LoS propagation and assume that the UEs are in the far-field of the array. Hence, we use \eqref{doc} with $\phi_k$ and  $\theta_k$ being the angles of the incident planar wave generated by UE $k=1,2$.} For convenience, we denote
\begin{align}\label{Z_AR_2_2}
\re\{{\bf Z}_{{\rm AR}}^{\rm ul}\}= \left(R_{\rm r}+ R_{\rm d}\right)\left[\begin{matrix}
1 & \mu \\
\mu & 1 
\end{matrix}\right]
\end{align}
where $|\mu| < 1$ accounts for the normalized mutual coupling between the two receiving antennas at the BS. The shape of $\mu$ as a function of the normalized antenna spacing $d_H/\lambda$ is reported in Fig.~\ref{fig:mu} for $R_{\rm r} = 73$\si{\ohm} and $R_{\rm d} = 10^{-3}R_{\rm r}$.
\vspace{-0.25 cm}
\subsection{Array gain}
The following result is found for the array gain, which is valid 
assuming a full matching network.
\begin{lemma} Consider the uplink with $M_{\rm BS}=2$. If a full matching network is used at the BS, then in single path LoS propagation the array gain (compared to a single antenna BS) for UE $k$ is
\begin{align}\label{eq:chanGain_simplified_1}
{\rm Array \, Gain} = 2 \frac{1 - \mu\cos \left(\psi_k\right)}{1 - \mu^2}
\end{align}
\textcolor{blue}{with $\psi_k = 2\pi \frac{d_H}{\lambda}\cos(\theta_k)\sin(\phi_k)$}. 
\end{lemma}
\begin{proof}
 In the case of full matching networks, the SNR $\gamma_k^{\rm ul}$ in \eqref{eq:sinr_ul} reduces to
  \begin{equation}
  \begin{split}
      \gamma_k^{\rm ul} & \triangleq \dfrac{p_k\left|\mathbf{u}_k^{\Htran} \mathbf{h}_k^{\rm ul}\right|^2}
  {  \mathbf{u}_k^{\Htran}\mathbf{R}_n^{\rm ul}\mathbf{u}_k} = \dfrac{p_k}{\frac{|Z_{\rm L}|^2\sigma^2}{|Z_{\rm L}+Z_{\rm opt}|^2}} \|\mathbf{h}_k^{\rm ul}\|^2 \\ &= \frac{p_k}{\sigma^2} \dfrac{|Z_{\rm L}+Z_{\rm opt}|^2}{|Z_{\rm L}|^2} |\xi_{\rm ul}|^2 {{\bf z}_{{\rm ART},k}^{{\rm ul},\Htran}}\re\{{\bf Z}^{\rm ul}_{{\rm AR}}\}^{-1} {\bf z}_{{\rm ART},k}^{\rm ul}
      \end{split}
  \end{equation}
as it follows from \eqref{d_i_ul_MR}. By using ${\bf z}_{{\rm ART},k}^{\rm ul} = Z_0 \alpha'(\theta_{k},\phi_{k},r_k) {\bf a}(\theta_{k},\phi_{k})$ (see Appendix A) and computing the inverse of \eqref{Z_AR_2_2} yields
\begin{equation}\label{snr_M_2_K_1}
      \gamma_k^{\rm ul} =  2A\frac{1 - \mu\cos \left(\psi_k\right)}{1 - \mu^2}.
  \end{equation}
with $A = |\alpha'(\theta_{k},\phi_{k},r_k)|^2 |\xi_{\rm ul}|^2 \frac{|Z_0|^2}{R_{\rm r}+ R_{\rm d}}\frac{|Z_{\rm L}+Z_{\rm opt}|^2}{|Z_{\rm L}|^2} \frac{p_k}{\sigma^2}$. 
The array gain is obtained after normalization with $A$.
\end{proof}

\begin{figure}
\centering
\begin{subfigure}{.45\textwidth}
  \centering
\begin{overpic}[width =\columnwidth]{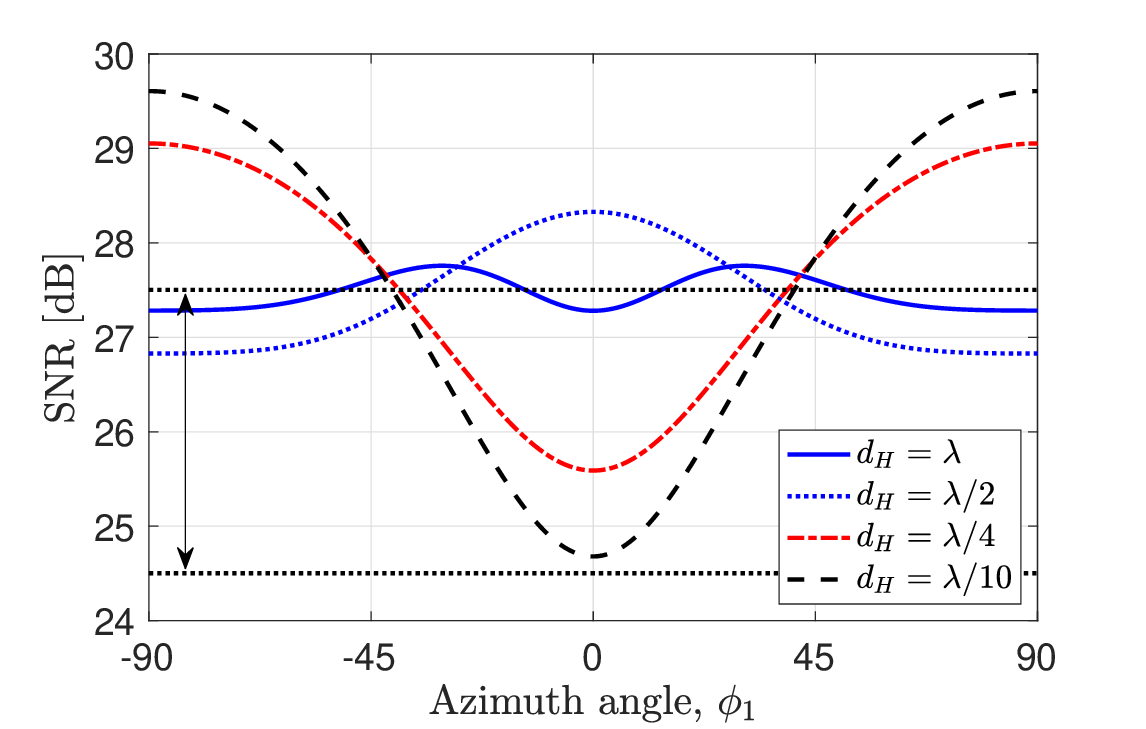}
 \put(17,26){\footnotesize $3$\,dB}
 \put(19,20){\footnotesize $M_{\rm BS}=1$}
     \put(23,19){\vector(1, -1){5}}
\end{overpic}
\caption{Full noise matching network}
\label{fig:sec5_fig4a}
\end{subfigure}
\begin{subfigure}{.45\textwidth}
  \centering
\begin{overpic}[width =\columnwidth]{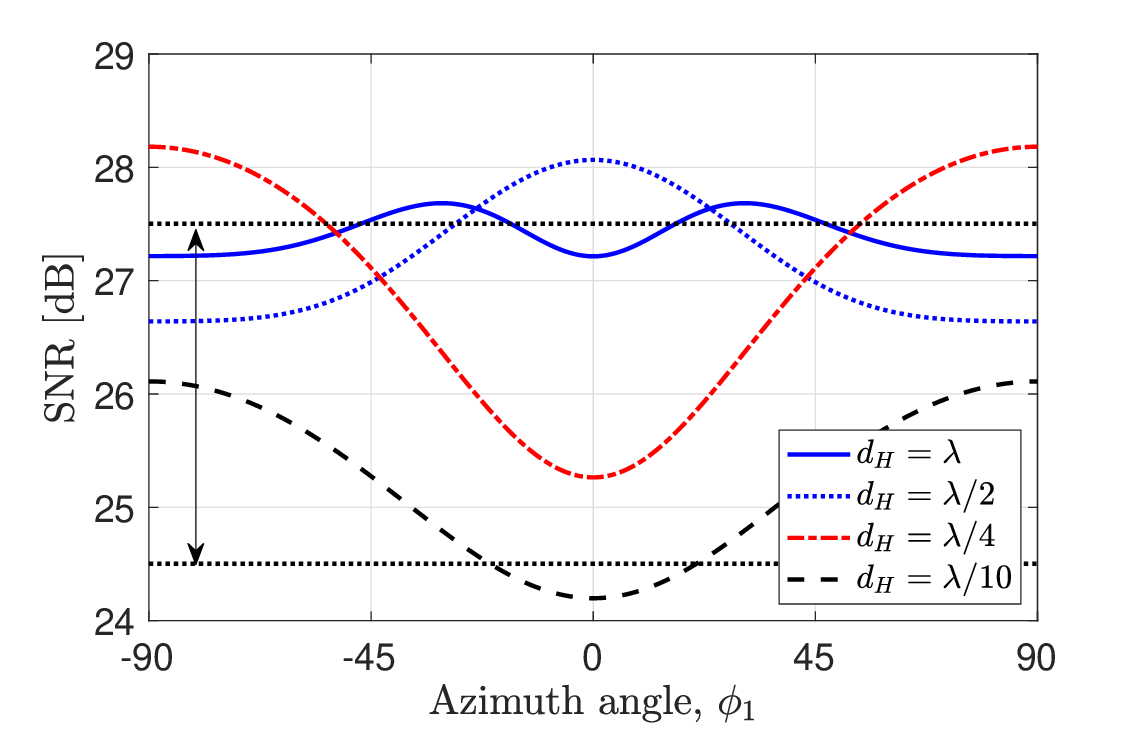}
 \put(18,31){\footnotesize $3$\,dB}
 \put(18,21){\footnotesize ${M_{\rm BS}}=1$}
      \put(23,20){\vector(1, -1){5}}
\end{overpic}
\caption{Self-impedance noise matching network}
\label{fig:sec5_fig4b}
\end{subfigure}
\begin{subfigure}{.45\textwidth}
  \centering
\begin{overpic}[width =\columnwidth]{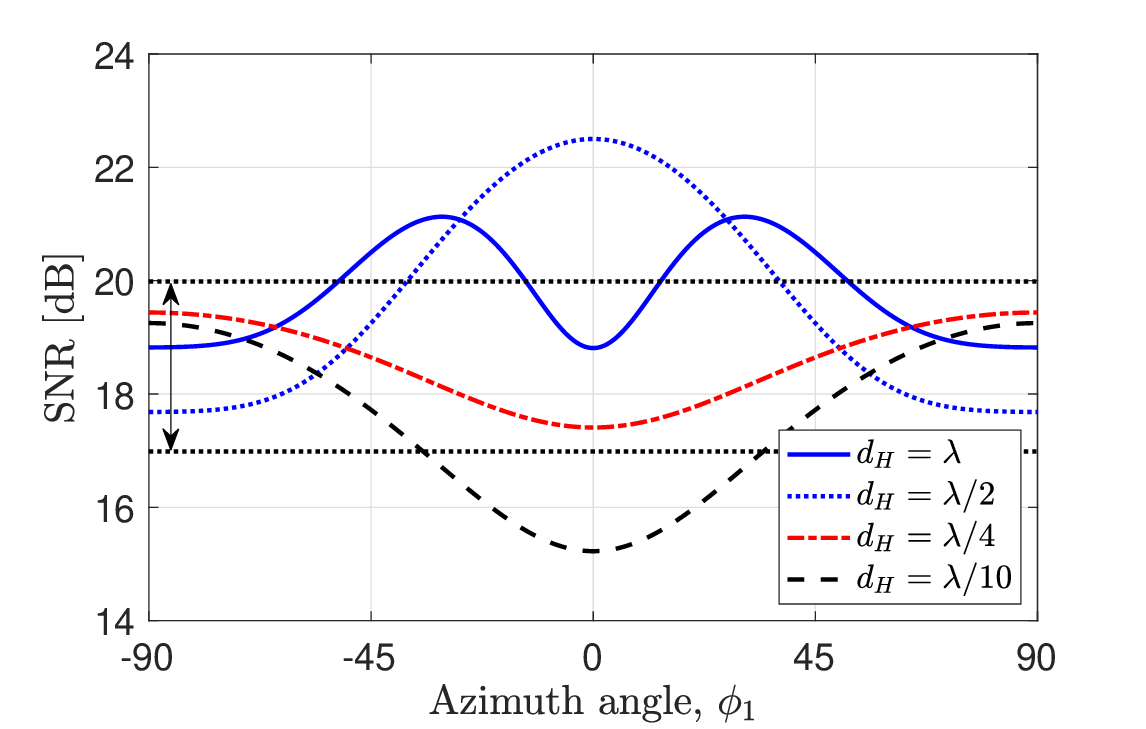}
 \put(16,29){\footnotesize $3$\,dB}
 \put(19,16){\footnotesize ${M_{\rm BS}}=1$}
      \put(23,19){\vector(1, 1){5}}
\end{overpic}
\caption{No noise matching network}
\label{fig:sec5_fig4c}
\end{subfigure}
\caption{ Uplink SNR of UE $1$ (in dB) with MMSE as $\phi_1$ varies with and without a noise matching network at the BS and ${M_{\rm BS}}=2$ antennas. The UE is located at a distance of $50$ meters and that the BS array is at an height of $10$ meters. Different values of $d_H$ are considered. The SNR with a single antenna is also reported as a benchmark.}\vspace{-0.5cm}
\label{fig:array_gain_vs_phi}
\end{figure}
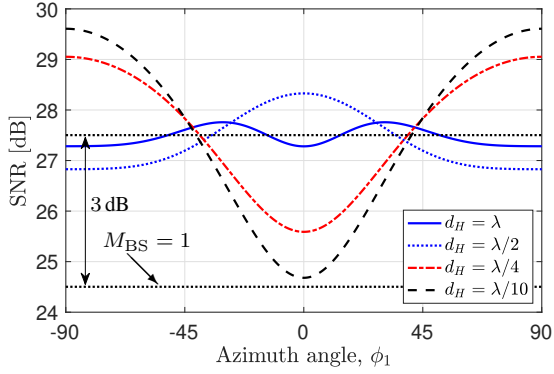
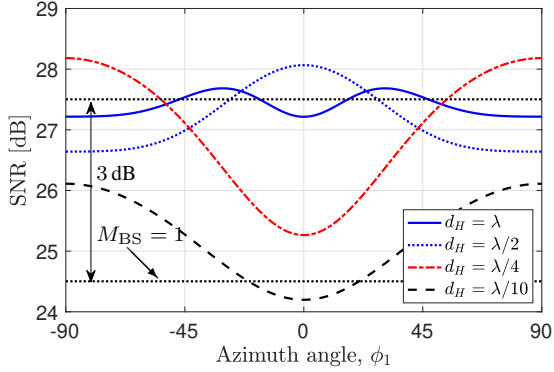
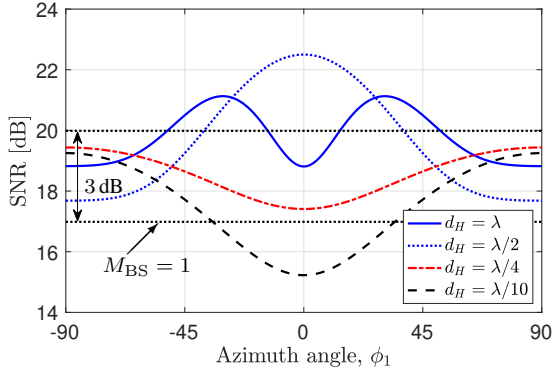

Lemma 1 shows that an array gain greater than $2$ can only be obtained if $\mu \ne 0$. From Fig.~\ref{fig:mu}, we see that $\mu$ can be positive or negative depending on $d_H/\lambda$, and the first null is at $d_H/\lambda \approx 0.43$. Particularly, $\mu > 0$ for $d_H/\lambda < 0.43$ while negative values are observed for $0.43 < d_H/\lambda < 1$. This has an important impact on the direction of arrival $(\theta_k,\phi_k)$ corresponding to the maximum value of the array gain. 
From~\eqref{eq:chanGain_simplified_1} it can be observed that for a fixed value of $d_H/\lambda$, the maximum value of the array gain is achieved when $\mu > 0$ and corresponds to the minimum value of $\cos \psi_k$ for $(\theta_k,\phi_k)$. On the other hand, if $\mu < 0$ the maximum is achieved for $(\theta_k,\phi_k)$ corresponding to the maximum value of $\cos \psi_k$. 
Assume for example $d_H/\lambda < 0.43$, which means $0 \le 2 \pi d_H/\lambda < 0.86 \pi < \pi$. Since $\mu > 0$, the maximum array gain is attained when $\cos \psi_k$ is minimum, i.e., when $\cos(\theta_k)\sin(\phi_k)=\pm 1$. This condition requires $\theta_k=0$ and $\phi_k=\pm \pi/2$, which represents the \textit{end-fire} direction of arrival. The corresponding maximum array gain is given by 
\begin{equation}
{\rm Maximum \, Array \, Gain} = 2\dfrac{ 1 - \mu \cos(2 \pi d_H / \lambda)}{1 - \mu^2}.    
\end{equation}
If $0.43 < d_H/\lambda < 1$, then $\mu <0$ and the maximum array gain is achieved when $\cos \psi_k$ is maximum, i.e., when $\cos(\theta_k)\sin(\phi_k)=0$. This requires $\phi_k = 0$ or $\theta_k = \pm \pi/2$. In particular, $\theta_k = 0$ and $\phi_k = 0$ corresponds to the \textit{front-fire} direction of arrival. In this case, we obtain 
\begin{equation}
{\rm Maximum \, Array \, Gain} = \dfrac{ 2 }{1 + \mu}.
\end{equation}

\begin{figure}[t!]
\centering
\begin{overpic}[width = 1\columnwidth]{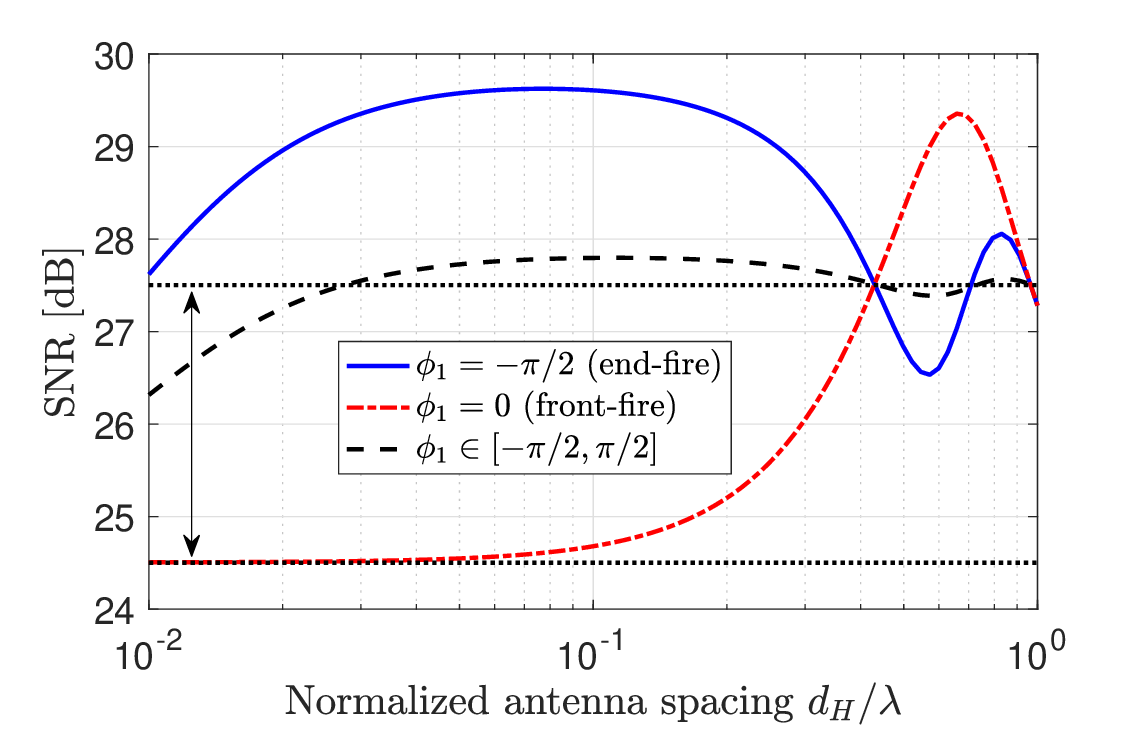}
 \put(18,25){\footnotesize $3$\,dB}
 \put(70,21){\footnotesize $M_{\rm BS}=1$}
      \put(74,20){\vector(1, -1){5}}
\end{overpic}
\caption{ Uplink SNR of UE $1$ (in dB) with MMSE with and without a noise matching network at the BS and $M_{\rm BS}=2$ antennas. The UE is located at a distance of $50$ meters and the BS array is at an height of $10$ meters. Different values of $\phi_1$ are considered. The SNR with a single antenna is also reported as a benchmark. \vspace{-0.5cm}}
\label{fig:SNR_spacing}
\end{figure}

Fig.~\ref{fig:sec5_fig4a} reports the SNR in dB for UE $1$ as a function of $\phi_1$ for different values of $d_H$ and with a \textit{full} matching network, i.e., ${\bf Z}_{\rm MR} = {\bf Z}^{\star}_{\rm MR}$. We assume that UE $1$ is located at a distance of $50$ meters and that the BS array is at an height of $10$ meters, which means $\theta_1 \approx -11^{\circ}$. The key parameters of the BS antenna array are reported in Table~\ref{tab:array_parameters}. For comparison, the SNR for the single-antenna case (i.e., $M_{\rm BS}=1$) is shown together with the line corresponding to an array gain of $3$ dB, i.e., the maximum array gain achievable with two uncoupled antennas. In agreement with the discussion above, the results of Fig.~\ref{fig:sec5_fig4a} show that, in the presence of a noise matching network, the array gain is maximum for $\phi_1 = \pm \pi/2$ (end-fire), when the antenna spacing $d_H$ is below $\lambda/4$ since $\mu>0$. On the contrary, it takes the maximum value for $\phi_1 = 0$ (front-fire) when $d_H=\lambda/2$ since $\mu<0$. For all the considered values of $d_H$, there exist ranges of $\phi_1$ for which the array gain is above $3$ dB. This proves that moving the antennas close to each other may have a positive effect that becomes negligible when $d_H$ is further reduced below $\lambda/10$. Interestingly, an array gain greater than $3$ dB can also be obtained for $d_H = \lambda/2$, when the transmitter is in front-fire. This is possible simply because $\mu \ne 0$ for $d_H = \lambda/2$, as shown in Fig.~\ref{fig:mu}.

The impact of matching network on the SNR when moving the antennas close to each other is illustrated in Fig.~\ref{fig:sec5_fig4b} and Fig.~\ref{fig:sec5_fig4c}, where we plot the SNR obtained with the self-impedance matching design (see Sect. II.K) and without a matching network. It can be observed that, for a fixed antenna spacing, the maxima and minima occur at the same values of $\phi_1$, regardless of the matching network design. However, the specific values of these maxima and minima are strongly influenced by the choice of the matching network. For instance, Fig.~\ref{fig:sec5_fig4b} demonstrates that reducing $d_H$ below $\lambda/4$ has a negative impact on both SNR and array gain. Furthermore, it is evident that the best performance, whether with a self-impedance matching network or without any matching network, is achieved when $d_H = \lambda/2$ and $\phi_1 = 0$.

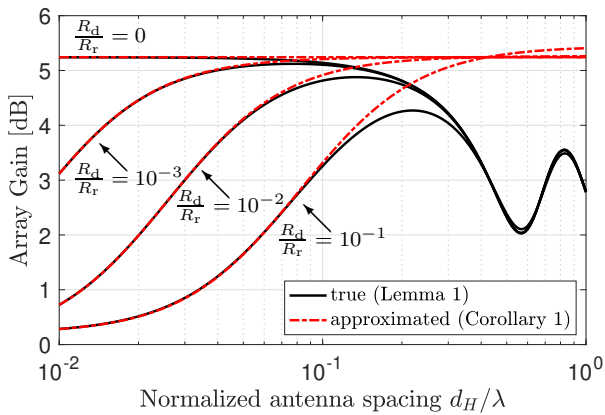
\begin{figure}[t!]
\centering
	\begin{overpic}[width = 1.02\columnwidth]{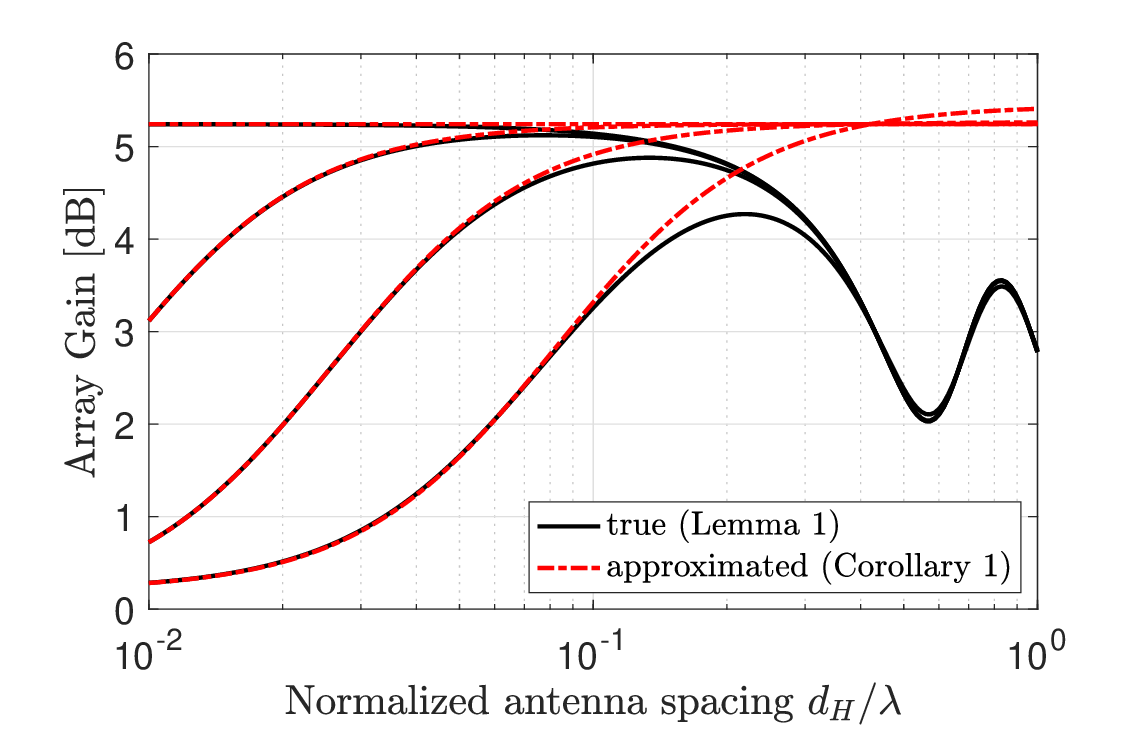}
 \put(15,55){\footnotesize $\frac{R_{\mathrm d}}{R_{\mathrm r}} = 0$}
	\put(15,34){\footnotesize $\frac{R_{\mathrm d}}{R_{\mathrm r}} = 10^{-3}$}
    \put(23,36){\vector(-1, 1){4}}
 \put(30,30){\footnotesize $\frac{R_{\mathrm d}}{R_{\mathrm r}} = 10^{-2}$}
 \put(38,32){\vector(-1, 1){4}}
 \put(45,25){\footnotesize $\frac{R_{\mathrm d}}{R_{\mathrm r}} = 10^{-1}$}
 \put(53,27){\vector(-1, 1){4}}
 \end{overpic}
\caption{ Array gain for different values of the dissipation resistance $R_{\rm d}$ when $\phi_1 = -\pi/2$ (end-fire). A noise matching network is used.}
\label{fig:ArrayGain_vs_Rd}
\end{figure}

To gain further insights into the effect of coupling as $\phi_1$ varies, Fig.~\ref{fig:SNR_spacing} plots the SNR of UE $1$ with a full noise matching network for $0.01 \le d_H/\lambda \le 1$. In particular, the black dashed curve has been obtained with $\phi_1$ uniformly distributed between $-\pi/2$ and $\pi/2$. The other parameters are the same as in Fig.~\ref{fig:sec5_fig4a}. The results are in agreement with those from Fig.~\ref{fig:sec5_fig4a}. Specifically, Fig.~\ref{fig:SNR_spacing} shows that in the presence of a noise matching network gains are achieved depending on the values of $\phi_1$. If $\phi_1$ is uniformly distributed between $-\pi/2$ and $\pi/2$, a minimal gain is achieved for $0.1 \le d_H/\lambda \le 1$ compared to uncoupled antennas. A loss is observed for small values of $d_H/\lambda$. This is a direct consequence of the dissipation resistance. To better understand this effect, the following corollary is given with $\mu_0 = \frac{R_{\rm r}}{R_{\rm r} + R_{\rm d}}$, $\mu_2 = \frac{\pi}{2} \frac{Z_0}{R_{\rm r} + R_{\rm d}}$ and $Z_0 = 377 \si{\ohm}$.
\begin{corollary} If ${d_{H}/\lambda \approx 0}$, then~\eqref{eq:chanGain_simplified_1} reduces to  
\begin{equation}\label{eq:asymp_array_gain}
2\frac{1 - \mu_0 + \left[2\mu_0\pi^{2} (\cos \theta_k \sin \phi_k)^{2} + \mu_2\right](d_{H}/\lambda)^{2}}{(1 + \mu_0)\left[1 - \mu_0 + \mu_2 (d_{H}/\lambda)^{2}\right]}.
\end{equation}
If ${d_{H}}/{\lambda} \to 0$, then~\eqref{eq:chanGain_simplified_1} tends to $\frac{2\mu_0\pi^{2} (\cos \theta_k \sin \phi_k)^{2} + \mu_2}{\mu_0\mu_2}$.
\end{corollary} 
\begin{IEEEproof}
With half-wavelength dipoles in side-by-side configuration, $\mu$ can be approximated as $\mu \approx \mu_{0}-\mu_{2} (d_{H}/\lambda)^{2}$ as $ d_{H}/\lambda \approx 0$. 
The Taylor expansion of $\cos \left(\psi_k\right)$ for $ d_{H}/\lambda \approx 0$ is $1-2 \pi^{2} (\cos\theta_k \sin\phi_k)^{2} (d_{H}/\lambda)^{2}$. Plugging these expressions into~\eqref{eq:chanGain_simplified_1} yields $\eqref{eq:asymp_array_gain}$ from which the asymptotic value for ${d_{H}}/{\lambda} \to 0$ follows.
\end{IEEEproof}

Fig.~\ref{fig:ArrayGain_vs_Rd} depicts the variations of~\eqref{eq:chanGain_simplified_1} and~\eqref{eq:asymp_array_gain} with respect to $d_H/\lambda$ for different values of the ratio $R_{\rm d}/R_{\rm r}$. The parameter $\phi_1$ is fixed at $-\pi/2$. The figure demonstrates the significant impact of $R_{\rm d}$ on the array gain. Specifically, as $R_{\rm d}/R_{\rm r}$ increases, the maximum array gain occurs at larger values of $d_H/\lambda$, while poor performance is observed when $d_H/\lambda$ approaches $0$.



\subsection{Interference gain}
The mutual coupling between antennas has also an impact on the interference term $\left|\mathbf{u}_k^{\Htran} \mathbf{h}_i^{\rm ul}\right|^2$ in~\eqref{eq:sinr_ul}. To show this, the following result is given for MR, i.e., $\mathbf{u}_k = \mathbf{h}_k^{\rm ul}$.


\begin{lemma} \label{lemma_3}Consider the uplink with MR and assume that $M_{\rm BS}=2$. If a full matching network is used at the BS, then in a single path LoS propagation scenario the interference gain (compared to a single antenna BS) between UEs $k$ and $i$ is
\begin{equation}
\begin{split}
{\rm Interf. \, Gain}  &= \dfrac{2\left[\cos(\frac{\psi_{k}-\psi_{i}}{2}) -\mu \cos(\frac{\psi_{k}+\psi_{i}}{2}\right]^2 }{(1-\mu \cos \psi_{k} )(1-\mu^{2})}\label{Gamma}
\end{split}
\end{equation}
\end{lemma}

\begin{IEEEproof}
With MR and full matching networks, the normalized interference term, as it follows from \eqref{d_i_ul_MR}, is 
\begin{equation}
\label{Interf1}
\frac{p_i\left|\mathbf{u}_k^{\Htran} \mathbf{h}_i^{\rm ul}\right|^2}
  {\mathbf{u}_k^{\Htran}\mathbf{R}_n^{\rm ul}\mathbf{u}_k} = A'\dfrac{|{\bf a}(\theta_{k},\phi_{k})^{H} \re\{{\bf Z}_{{\rm AR}}^{\rm ul}\}^{-1} {\bf a}(\theta_{i},\phi_{i})|^{2}}{{\bf a}(\theta_{k},\phi_{k})^{H} \re\{{\bf Z}_{{\rm AR}}^{\rm ul}\}^{-1} {\bf a}(\theta_{k},\phi_{k})}
\end{equation}
with $A' = \frac{p_i|\alpha'(\theta_{i},\phi_{i},r_i)|^2}{\sigma^2} \frac{|\xi_{\rm ul}|^2 |Z_0|^2 |Z_{\rm L}+Z_{\rm opt}|^2}{|Z_{\rm L}|^2} $. From \eqref{Z_AR_2_2}, $|{\bf a}(\theta_{k},\phi_{k})^{H} \re\{{\bf Z}_{{\rm AR}}^{\rm ul}\}^{-1} {\bf a}(\theta_{i},\phi_{i})|^2$ is obtained as
\begin{equation}
\begin{split}
\dfrac{4 \left[\cos(\psi_{k}/2-\psi_{i}/2) -\mu \cos(\psi_{k}/2+\psi_{i}/2)\right]^2}{(R_{r}+R_{d})^{2}(1-\mu^{2})^{2}}
\end{split}
\end{equation}
whereas ${\bf a}(\theta_{k},\phi_{k})^{H} \re\{{\bf Z}_{{\rm AR}}^{\rm ul}\}^{-1} {\bf a}(\theta_{k},\phi_{k})$ is given by
\begin{equation}
\frac{2}{R_{\rm r}+ R_{\rm d}}\frac{1 - \mu\cos \left(\psi_k\right)}{1 - \mu^2}.
\end{equation}
The result in \eqref{Gamma} follows after normalization with $A'/(R_{\rm r}+ R_{\rm d})$. 
\end{IEEEproof}

\begin{figure}
\centering
\begin{subfigure}{.5\textwidth}
  \centering
\begin{overpic}[width = \columnwidth]{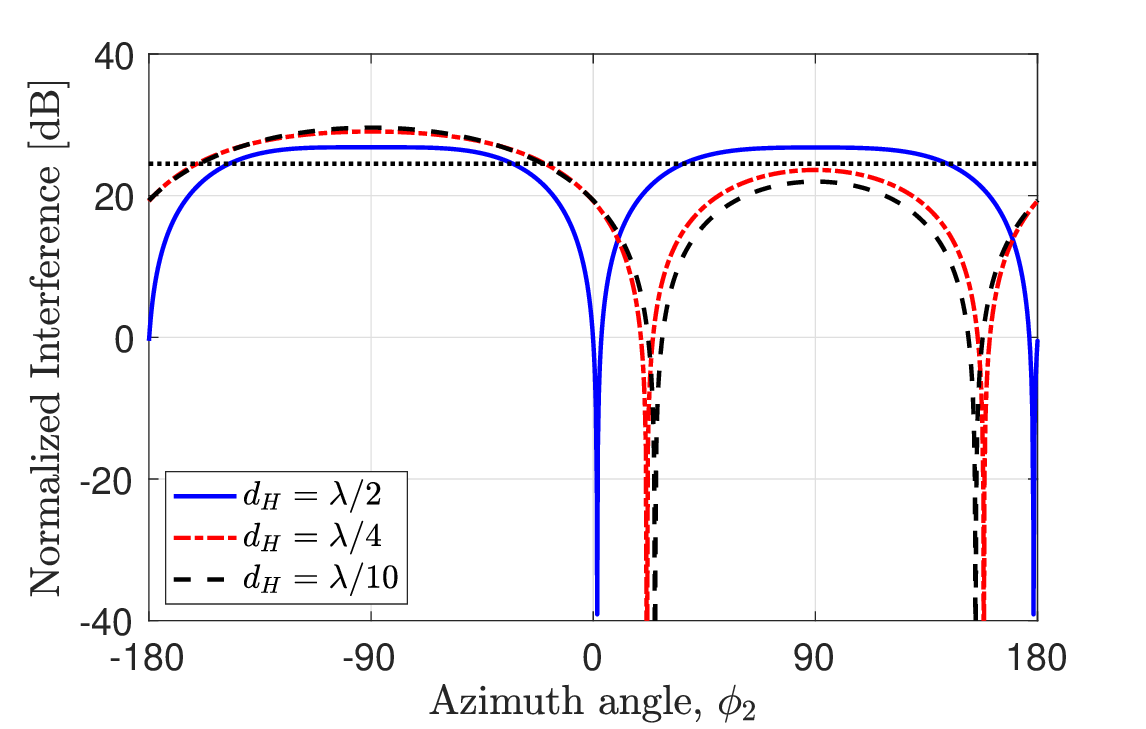}
 \put(19,42){\footnotesize $M_{\rm BS}=1$}
     \put(23,44){\vector(1, 1){5}}
\end{overpic}
\caption{$\phi_1 = -\pi/2$ (end-fire)}
\label{fig:sec5_fig6a}
\end{subfigure}
\begin{subfigure}{.5\textwidth}
  \centering
\begin{overpic}[width = \columnwidth]{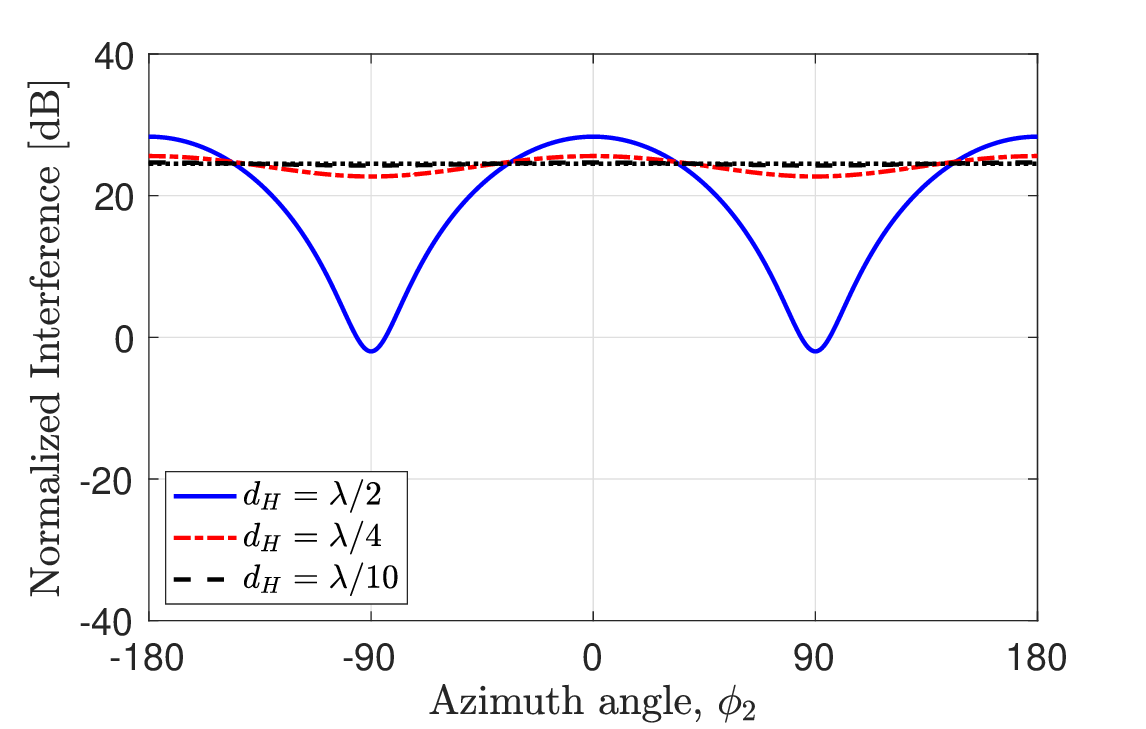}
 \put(19,55){\footnotesize $M_{\rm BS}=1$}
      \put(23,55){\vector(1, -1){5}}
\end{overpic}
\caption{$\phi_1 = 0$ (front-fire)}
\label{fig:sec5_fig6c}
\end{subfigure}
\caption{ Normalized interference of UE $1$ (in dB) with MR as $\phi_2$ varies with a full noise matching network at the BS and $M_{\rm BS}=2$ antennas. Different values of $d_H$ are considered. The single antenna case is also reported as a benchmark.}\vspace{-0.5 cm}
\label{fig:Interference_vs_phi_M2K2}
\end{figure}

Similarly to the array gain, the interference gain is also influenced by the parameters ${d_H, (\theta_k, \phi_k)}$ and ${d_H, (\theta_i, \phi_i)}$ through $\psi_k$ and $\psi_i$. The expression for the interference gain is more complex, making it challenging to gain direct insights into the interplay of these parameters. However, the numerical results shown in Fig.~\ref{fig:Interference_vs_phi_M2K2} reveal that the coupling effects observed with densely spaced antennas can either enhance or hinder the interference rejection capabilities of MR (but the same considerations apply to MMSE), depending on the directions of arrival of the interfering signal. For example, assuming that UE 1 is in end-fire (as in Fig.~\ref{fig:sec5_fig6a}) the best performance is observed with $\phi_2=0^{\circ}$ when $d_H = \lambda/2$ and for $\phi_2 \approx 24^{\circ}$ when $d_H = \lambda/10$. It is worth observing that, when $\phi_1 = 0^{\circ}$, moving the antennas close to each other has minimal effects on interference rejection, as shown in Fig.~\ref{fig:sec5_fig6c}. In this case, the best performance is obtained with $d_H = \lambda/2$ and $\phi_2 \pm 90^{\circ}$. 

\begin{figure}
\centering
\begin{subfigure}{.5\textwidth}
  \centering
\begin{overpic}[width = \columnwidth]{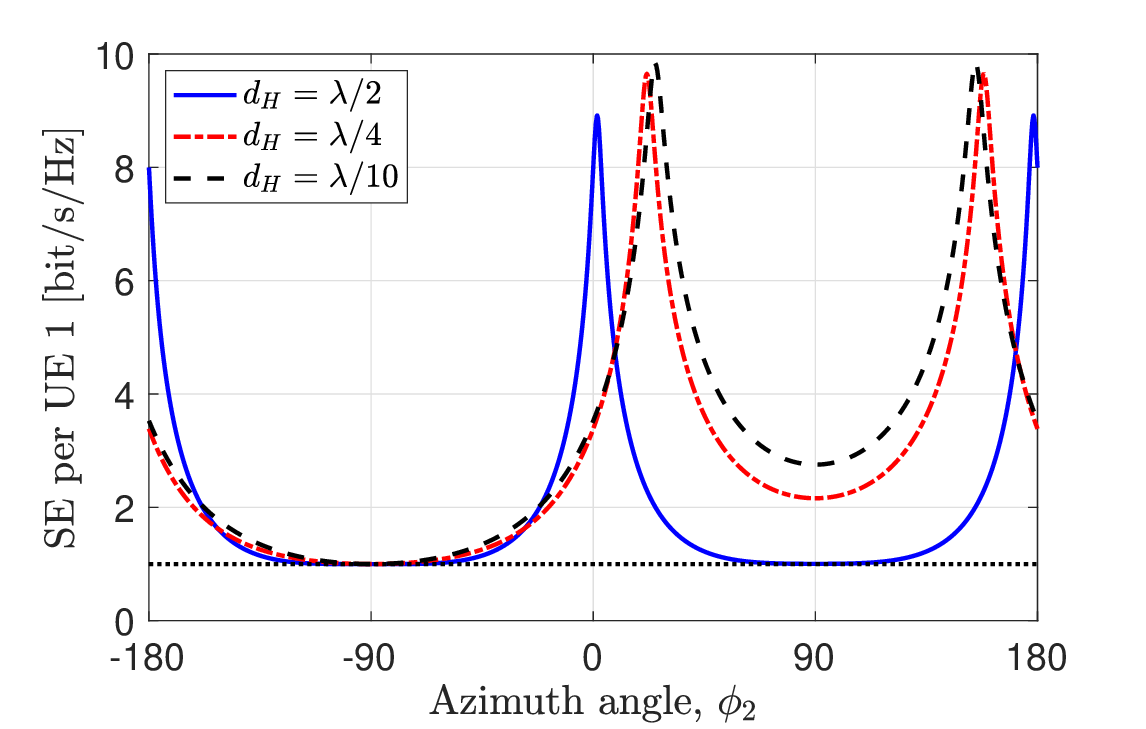}
 \put(23,21){\footnotesize $M_{\rm BS}=1$}
      \put(27,20){\vector(1, -1){5}}
\end{overpic}
\caption{$\phi_1 = -\pi/2$ (end-fire)}
\label{fig:sec5_fig6b}
\end{subfigure}
\hspace{0.2cm}
\begin{subfigure}{.5\textwidth}
  \centering
\begin{overpic}[width = \columnwidth]{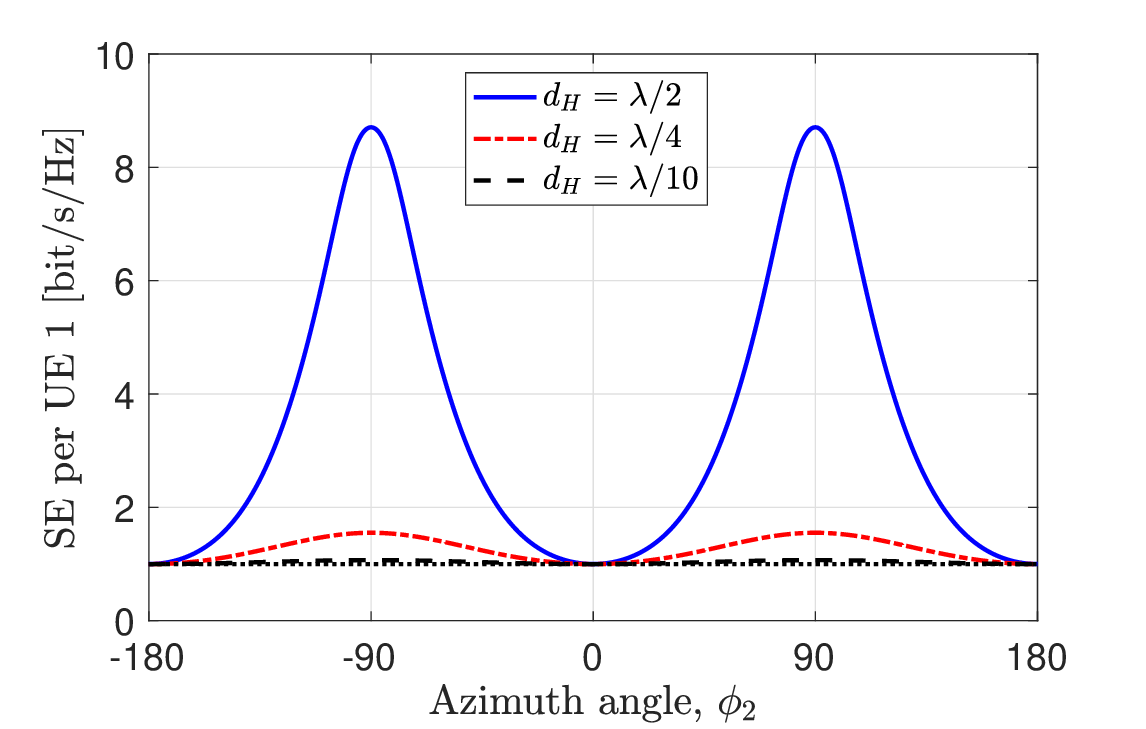}
 \put(23,21){\footnotesize $M_{\rm BS}=1$}
      \put(27,20){\vector(1, -1){5}}
\end{overpic}
\caption{$\phi_1 = 0$ (front-fire)}
\label{fig:sec5_fig6d}
\end{subfigure}
\caption{ SE of UE $1$ (in dB) with MR as $\phi_2$ varies with a full noise matching network at the BS, $M_{\rm BS}=2$ antennas, and different values of $d_H$. The single antenna case is also reported.\vspace{-0.5cm}}
\label{fig:SE_vs_phi_M2K2}
\end{figure}

\subsection{Spectral efficiency}
Both array and interference gains contribute to the overall SINR and ultimately impact the spectral efficiency of the different UEs. To quantify this, Fig.~\ref{fig:SE_vs_phi_M2K2} plots 
the spectral efficiency  of UE 1 as a function of $\phi_2$ in the same simulation scenario of Fig.~\ref{fig:Interference_vs_phi_M2K2}. The single antenna case is also reported as a benchmark. It is assumed the same transmit power for the two users. The results in Figs.~\ref{fig:sec5_fig6b}-\ref{fig:sec5_fig6d} can easily be explained with those in Fig.~\ref{fig:SNR_spacing} and Figs.~\ref{fig:sec5_fig6a}-\ref{fig:sec5_fig6c}. In particular, as expected, the points of minimum/maximum in Figs.~\ref{fig:sec5_fig6b} and \ref{fig:sec5_fig6d} correspond to the points of maximum/minimum in Figs.~\ref{fig:sec5_fig6a} and \ref{fig:sec5_fig6c}. When $\phi_1 = -90^{\circ}$ (as in Fig.~\ref{fig:sec5_fig6b}), the maximum SE (more than four times larger compared to the single antenna case) is achieved with $d_H = \lambda/10$ (when $\phi_2 \approx 24^{\circ}$) but significant gains are also observed for $d_H = \lambda/2$ when $\phi_2 = 0^{\circ}$. On the other hand, when $\phi_1 = 0^{\circ}$ poor performance is obtained by moving the antennas close to each other, and $d_H=\lambda/2$ is the best option.

Fig.~\ref{fig:Uplni_SE_with_K_2} plots the average SE per UE with different matching networks at BS: \textcolor{blue}{\emph{i}) Full Matching Network (Full MN); \emph{ii}) Self-Impedance Matching Network (SI MN); and \emph{iii}) No Matching Network (No MN).} 
We assume that the angles of arrival $\phi_1$ and $\phi_2$ for the two UEs are within the range $[-\pi/2, \pi/2]$ and that both UEs are located at a distance of $50$ meters.
The results indicate that decreasing $d_H/\lambda$ has a detrimental impact on the spectral efficiency, regardless of the matching network employed. Additionally, it is evident that the best performance is achieved when $d_H/\lambda \geq 0.5$, meaning that there is no significant advantage in using a full matching network compared to the self-impedance matching design. As anticipated, a considerable reduction in spectral efficiency is observed when no noise matching network is utilized.

\begin{figure}[t!]
\centering
\includegraphics[width = 1.02\columnwidth]{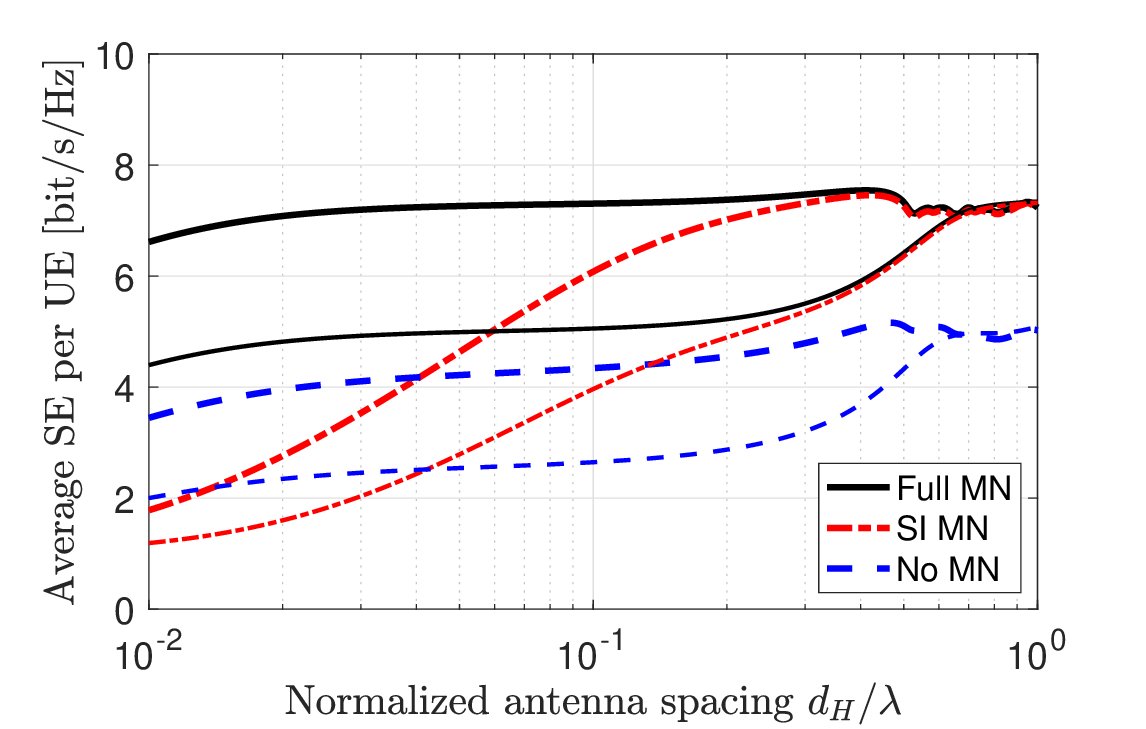}
\caption{ Average uplink SE per UE with MMSE, $M_{\rm BS} = 2$ antennas and $K=2$. The ticker lines are obtained with $\phi_1,\phi_2 \in [-\pi/2,\pi/2]$, while the others $\phi_1,\phi_2 \in [-\pi,0]$. Different noise matching networks at the BS are considered. \vspace{-0.5cm}}
\label{fig:Uplni_SE_with_K_2}
\end{figure}




\begin{figure}
\centering
\begin{subfigure}{.5\textwidth}
\begin{overpic}[width = \columnwidth]{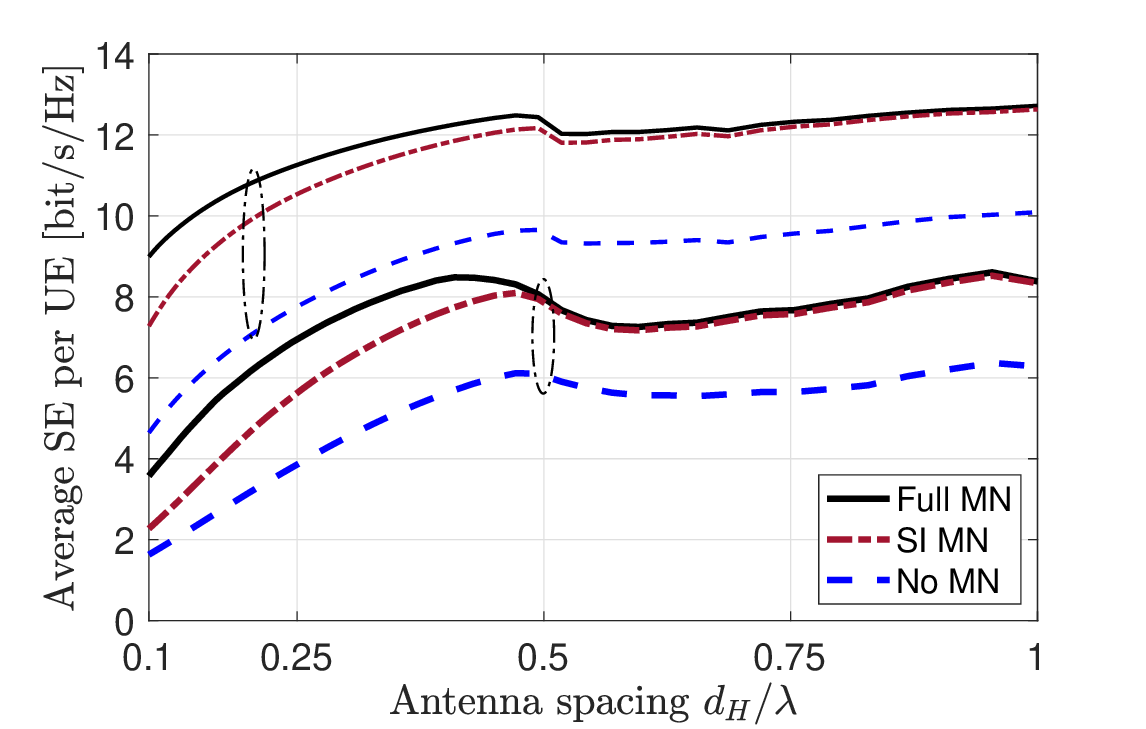}
 \put(14,54){\footnotesize $M_{\rm BS}=64$}
      \put(17,53){\vector(1, -1){4}}
       \put(49,23){\footnotesize $M_{\rm BS}=16$}
      \put(52,26){\vector(-1, 1){4}}
\end{overpic}
\caption{$K=10$ and $M_{\rm BS}=16,64$}
\label{fig:SE_vs_dH_K10_M16,64}
\end{subfigure}
\begin{subfigure}{.5\textwidth}
  \centering
\begin{overpic}[width = \columnwidth]{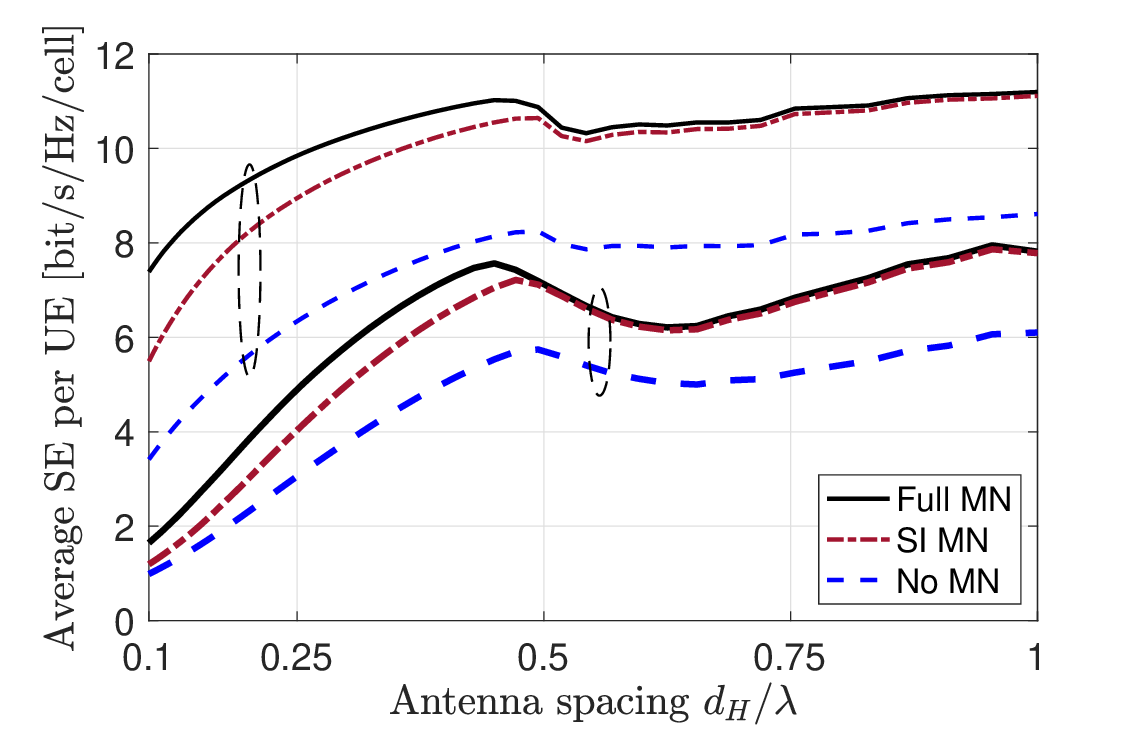}
 \put(14,54){\footnotesize $K=8$}
      \put(17,53){\vector(1, -1){4}}
       \put(54,22){\footnotesize $K=24$}
      \put(57,26){\vector(-1, 1){4}}
\end{overpic}
\caption{$M_{\rm BS}=32$ and $K=8,24$}
\label{fig:SE_vs_dH_M32_K8,24}
\end{subfigure}
\caption{ SE per UE with MMSE and azimuth angles $\{\phi_k; k =1,\ldots,K\}$ of UEs uniformly distributed within the sector $ [-\pi/2, \pi/2]$. The thicker lines correspond to $M_{\rm BS}=16$ in Fig.~\ref{fig:NMSE}a and to $K=24$ in Fig.~\ref{fig:NMSE}b.}
\label{fig:NMSE}
\end{figure}
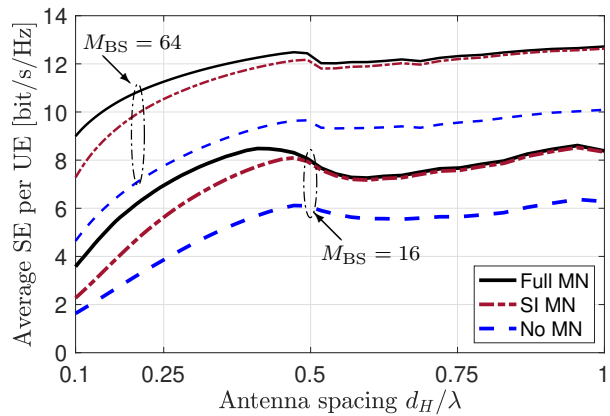
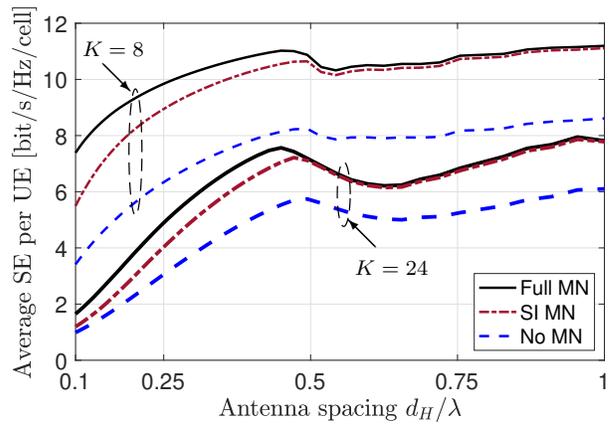
\section{Numerical Analysis}
The analysis presented above highlights that the mutual coupling effects resulting from closely spaced antennas can potentially provide benefits to the uplink spectral efficiency in single-user and multi-user Holographic MIMO systems, depending on the specific propagation conditions and \textcolor{blue}{impedance matching networks used}.
The analysis focused on a simplified uplink case study with two antennas and two UEs. Next, the numerical analysis is expanded to more realistic scenarios, including a larger number of antennas, \textcolor{blue}{arranged side-by-side in a uniform linear array}, and multiple UEs. Additionally, the analysis considers the case of densely packed antennas in a space-constrained form factor. By exploring these scenarios, a more comprehensive understanding of the benefits of mutual coupling in Holographic MIMO systems can be obtained.


The system parameters are those reported in Table~\ref{tab:array_parameters}. \textcolor{blue}{We consider a scenario with single-path LoS propagation and model the wireless channel as in \eqref{eq:d_OC}.} The BS is positioned at a height of $10$~m. The azimuth angle of each UE is randomly distributed within the sector $[-\pi/2, \pi/2]$ while the elevation angle depends on the distance from the BS. UEs are randomly dropped at a minimum distance of $15$~m and a maximum distance of $150$~m from the BS, and they transmit with the same power. The results are obtained by averaging over $1000$ UE drops. \textcolor{blue}{While we assume impedance matching is consistently applied at the UE, the three different case, namely, Full Matching Network (Full MN), Self-Impedance Matching Network (SI MN) and No Matching Network (No MN), are considered for the BS.}

Due to space limitations, our main emphasis is on the uplink but we put a specific focus on addressing the duality implication in the downlink. Although, we focus LoS propagation, similar results can be obtained with different channel models, e.g., based on stochastic approaches.

\subsection{Fixing the Number of Antennas while Varying Array Size}

Fig.~\ref{fig:SE_vs_dH_K10_M16,64} illustrates the average SE per UE in the uplink as a function of $d_H/\lambda$ for two different antenna configurations: $M_{\rm BS}=16$ and $M_{\rm BS}=64$, with a fixed number of UEs, $K=10$. The results show that, \textit{when the number of antennas is fixed, reducing the antenna spacing generally has a negative impact on the average SE}. Better performance is observed for $d_H/\lambda > 0.5$. In this range, employing a full matching network yields only a marginal gain compared to the self-impedance matching design. However, a significant decrease in SE occurs when no matching network is utilized. Additionally, as expected, increasing the number of antennas results in higher SE due to enhanced interference rejection capabilities. Similar conclusions can be drawn from Fig.~\ref{fig:SE_vs_dH_M32_K8,24}, where the number of antennas is fixed at $M_{\rm BS}=32$, while the number of UEs is varied between $K=10$ and $K=30$.
The curves in Fig.~\ref{fig:SE_vs_dH_K10_M16,64} were obtained with $R_{\rm d} =10^{-3} R_{\rm r}$. Numerical results (not reported for space limitations) show that the SE worsens as $R_{\rm d}$ increases but similar behaviors can be observed. We also notice that the effect of $R_{\rm d}$ is more significant for $d_H < \lambda/2$ while a marginal impact is observed for large antenna spacings (i.e, $d_H > \lambda/2$).

\begin{figure}
\centering
\begin{overpic}[width = 1\columnwidth]{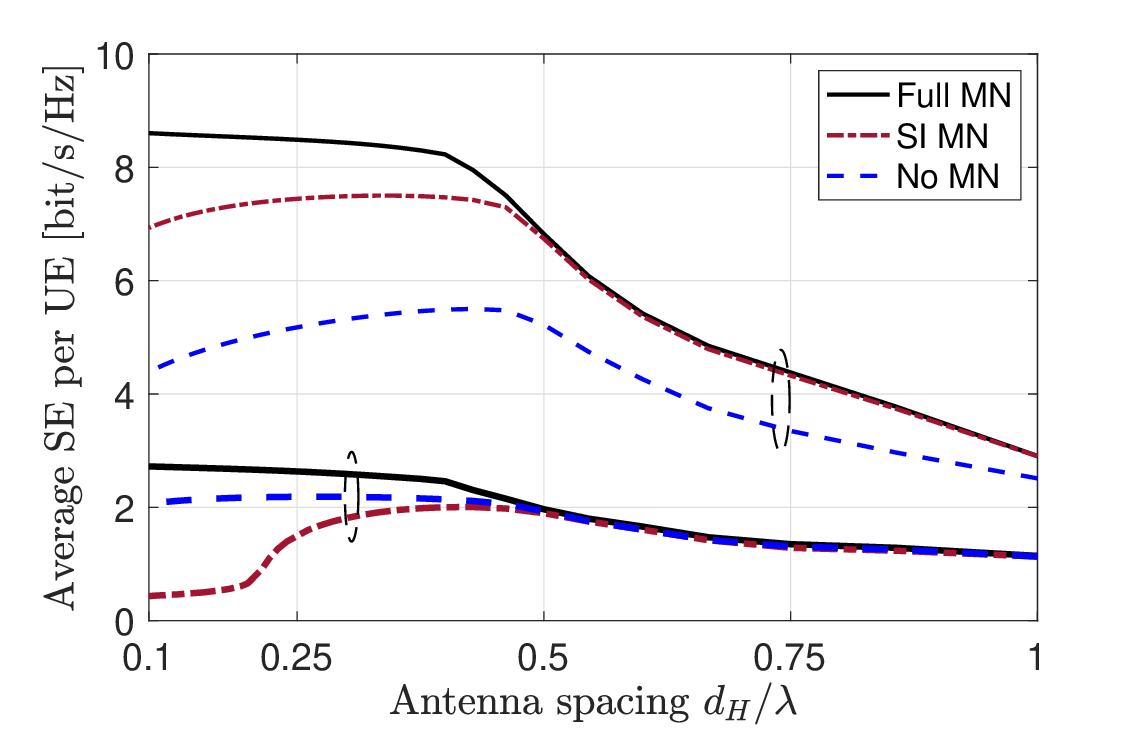}
 \put(71,39){\footnotesize MMSE}
      \put(73,38){\vector(-1, -1){4}}
       \put(33,10){\footnotesize MR}
      \put(35,13){\vector(-1, 1){4}}
\end{overpic}
\caption{  Average SE per UE when the ULA has a size of $6\lambda$. The number of UEs is $K=10$ with $\theta_k \in [-\pi/2,\pi/2]$ for $k=1,\ldots,K$. The thicker lines are obtained with MR combining while the other ones with MMSE.}
\label{fig:NMSE_11}\vspace{-0.3cm}
\end{figure}

\begin{figure}[t!]
\centering
\begin{overpic}[width = 1\columnwidth]{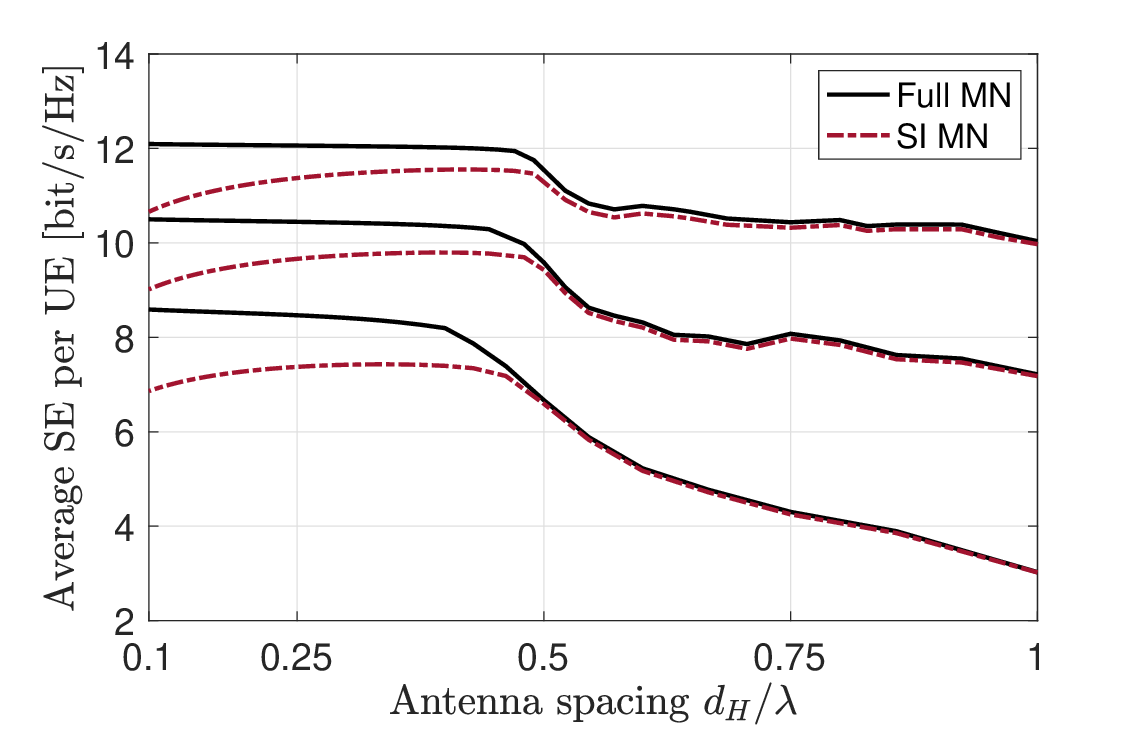}
 \put(55,13){\footnotesize $L_H = 6 \lambda$}
      \put(59,15){\vector(1, 1){4}}
 \put(72,25){\footnotesize $L_H = 12 \lambda$}
      \put(76,27){\vector(1, 1){4}}
       \put(55,51){\footnotesize $L_H = 24 \lambda$}
      \put(62,50){\vector(1, -1){4}}
\end{overpic}
\caption{ Average SE per UE with full noise or SI matching, for $K=10$ and an MMSE combiner.}\vspace{-0.5cm}
\label{SE_vs_dH_K10_LH61224}
\end{figure}
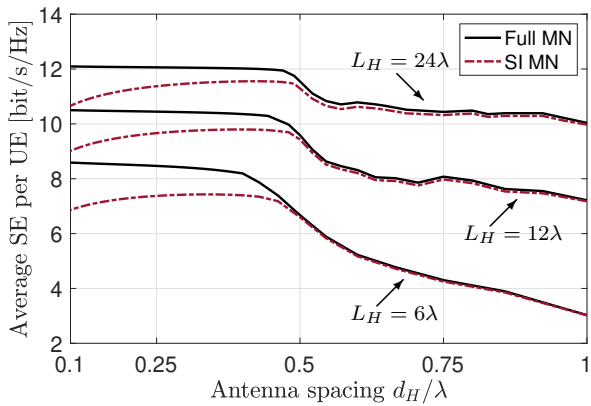

\subsection{Fixing the Array Size while Varying the Number of Antennas}

Fig.~\ref{fig:NMSE_11} illustrates the average SE per UE in the uplink as a function of $d_H/\lambda$ for a fixed array size $L_H = 6 \lambda$. Both MR and MMSE receivers exhibit similar SE behaviors. Notably, when a full matching network is employed, SE increases as $d_H/\lambda$ decreases due to the augmented number of antennas $M_{\rm BS}=L_H/d_H + 1$. This increase in antennas contributes to higher array gain and improved interference rejection. However, without a matching network or with a self-impedance matching network, the optimal performance is achieved when $d_H/\lambda \approx 0.4$. Going below this value may result in a decrease in SE. \textcolor{blue}{It is important to emphasize that the SE improvement observed when reducing $d_H$ with a full matching network cannot be attributed to antenna coupling. This is evident from the declining trend of the SE in Fig.~\ref{fig:NMSE} as the antenna spacing decreases.}

Fig.~\ref{SE_vs_dH_K10_LH61224} shows the average SE per UE as a function of $d_H/\lambda$ for three different values of the array size $L_H$. The number of UEs is $K=10$ and an MMSE combiner is employed, with either full noise or SI matching networks. As can be seen from the results, the behavior is the same irrespective of the array size. We only observe that, moving the antennas close to each other, the gain reduces as $L_H$ increases. In particular, when $L_H=6 \lambda$ and a full matching network is used, for $d_H=\lambda/10$ the average SE is about $4 \operatorname{bit/s/Hz}$ and drops to about $2 \operatorname{bit/s/Hz}$ for $d_H=\lambda$, with a ratio of 2 between the two values. On the other hand, when $L_H=24 \lambda$ the ratio decreases to about $6/4.5 \approx 1.33$. \textcolor{blue}{This suggests that the benefits of densely packing antennas are more pronounced for smaller array sizes. In addition, we find that the benefits of increasing the number of array elements for a given array size (resulting in a continuous antenna in practice) gradually diminish beyond a certain threshold.}   



\begin{figure}
\centering
\begin{subfigure}{.5\textwidth}
\begin{overpic}[width = \columnwidth]{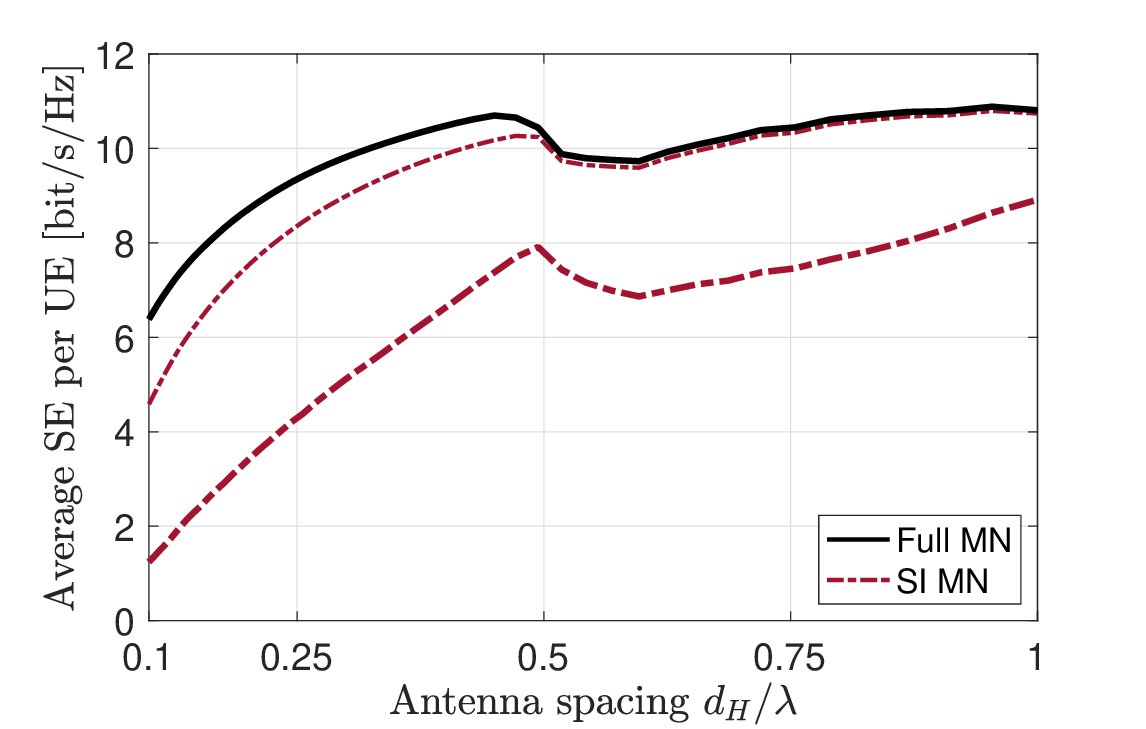}
\end{overpic}
\caption{$M_{\rm BS}= 32$ and $K=10$.}
\label{fig:SE_vs_dH_M32K8_DownLink}
\end{subfigure}
\begin{subfigure}{.5\textwidth}
  \centering
\begin{overpic}[width = \columnwidth]{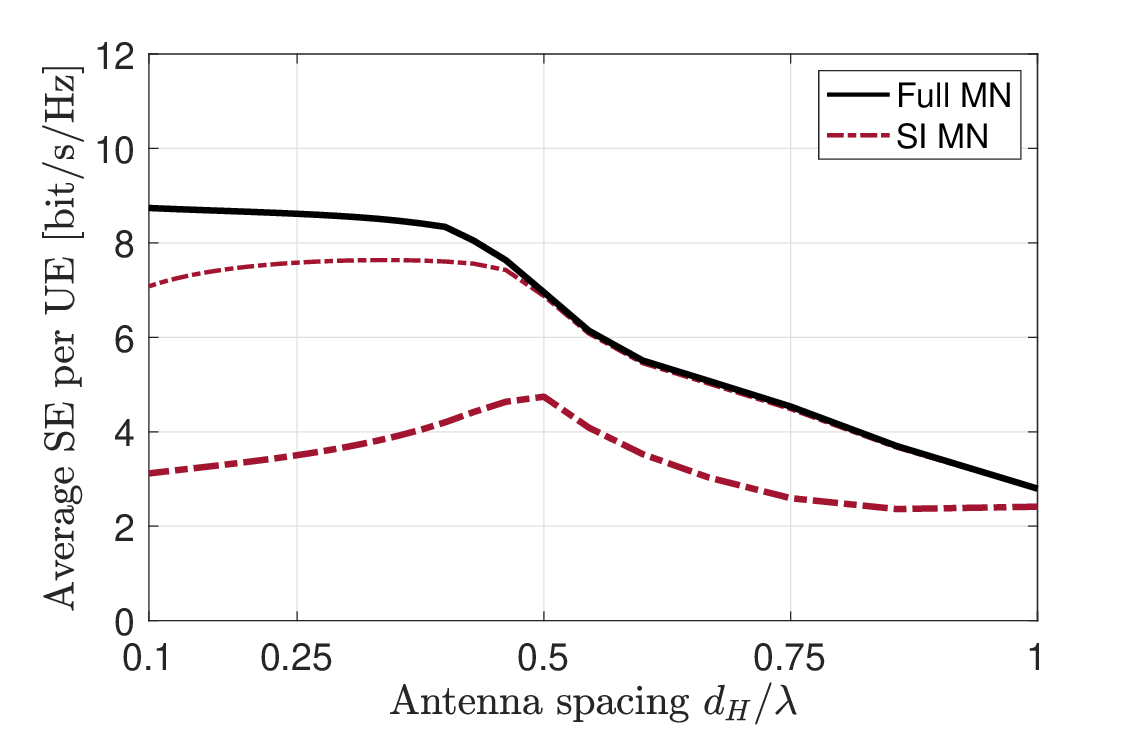}
\end{overpic}
\caption{$L_{H}= 6\lambda$ and $K=10$.}
\label{fig:SE_vs_dH_Lh6K10_DownLink}
\end{subfigure}
\caption{ Average SE per UE in the downlink with MMSE in the same setup of Figs. 10(b) and 12. The thicker lines refer to the case in which the uplink channels $\{\mathbf{h}_k^{\rm ul}; k=1,\ldots,K\}$ (instead of downlink channels $\{\mathbf{h}_k^{\rm dl}; k=1,\ldots,K\}$) are used for the computation of the precoding vectors. Both cases with a fixed number of antennas or a fixed antenna array aperture are considered.}
\label{fig:NMSE}
\end{figure}

\subsection{Impact of uplink and downlink duality}
We now consider the downlink with MMSE precoding, and with either a full or an SI matching network. From Table~\ref{tab:duality} it is seen that, with a full matching network, ${\bf h}_k^{\rm ul}$ and ${\bf h}_k^{\rm dl}$ differ only for a scaling factor. Accordingly, it is correct to design the MMSE precoder in downlink by using the measured value of ${\bf h}_k^{\rm ul}$ in uplink. On the other hand, when an SI matching network is employed, the uplink-downlink channel duality requires to apply a linear transformation to ${\bf h}_k^{\rm ul}$. A performance loss is incurred if this is not done.

Fig.~\ref{fig:SE_vs_dH_M32K8_DownLink} shows the average SE per UE in the same setup of Fig.~\ref{fig:SE_vs_dH_M32_K8,24}, i.e., with $M_{\rm BS} = 32$ antennas and $K = 10$ UEs. We see that, with a full matched network, the performance in uplink and downlink is the same. As for the SI matching design, two different cases have been considered. In the first case, the MMSE precoder is computed by using ${\bf h}_k^{\rm ul}$ instead of ${\bf h}_k^{\rm dl} = \dfrac {\alpha_{\rm dl}} {\alpha_{\rm ul}} {\bf B}_{\rm dl}^{-\Ttran/2}{\bf A}_{\rm dl,ul}{\bf h}_k^{\rm ul}$, as indicated in Table~\ref{tab:duality}. In the second case, the MMSE precoder is correctly computed taking the matrix ${\bf B}_{\rm dl}^{-\Ttran/2}{\bf A}_{\rm dl,ul}$ into account. Comparing the results in Fig.~\ref{fig:SE_vs_dH_M32K8_DownLink} with those in 
Fig.~\ref{fig:SE_vs_dH_M32_K8,24}, we see that in the latter case the average SE is the same in uplink and downlink, while a considerable loss is observed in the former case (thicker line), especially at low values of $d_H/\lambda$. The same conclusions can be drawn from Fig.~\ref{fig:SE_vs_dH_Lh6K10_DownLink}, obtained in the simulation setting of Fig.~\ref{SE_vs_dH_K10_LH61224}, which shows the average SE per UE for a fixed size $L_H = 6 \lambda$ of the array and $K=10$.

\section{Conclusions}\label{sec:conclusions}

Building on the multiport communication theory (e.g.,~\cite{Nossek2014,Nossek2010}), a physically-consistent representation of MIMO channels can be derived and directly used by communication theorists as a baseline for modelling the uplink and downlink of holographic MIMO communications, with closely spaced antennas. Particularly, we used it to study the effects of mutual coupling on the spectral efficiency and to gain insights into interplay between antenna spacing and impedance matching network designs. \textcolor{blue}{We focused on side-by-side half-wavelength dipoles in a LoS scenario.} Numerical and analytical results showed that for the investigated scenarios, a fixed number of closely spaced antennas with impedance matching can provide spectral efficiency benefits, but only for specific directions of the incident signals. On average, the spectral efficiency gains may be marginal or even non-existent. We explored a scenario where antennas were closely packed in a space-constrained form, and we showed that reducing the antenna distance led to an increase in spectral efficiency. However, this increase becomes negligible as the array aperture size grows (in the order of tens of wavelengths). The uplink and duality duality was also investigated for different matching network designs. \textcolor{blue}{We limited our study to uniform linear arrays of side-by-side half-wavelength dipoles operating in LoS conditions. However, we notice that the analytical framework can be used to extend the results to different array configurations and non-LoS propagation conditions.}


{\color{blue}\section*{Appendix A}\label{Appendix}
In this Appendix, we derive the expression of ${\bf d}_{{\rm OC}}$ given in \eqref{eq:d_OC}. The computation will be made for a LoS propagation scenario on the basis of the two following assumptions:

\begin{assumption}
Each antenna of the BS array is in the far field of the UE transmit antenna. In the far-field region, the transmit antenna behaves like a \textit{source point}, so the radiated field can be approximately characterized by spherical wavefronts. If the transmitting antenna has a maximum dimension of $L_t$, the far-field region is at distances greater than $2L_t^2/\lambda$.

\end{assumption}
\begin{assumption}
The electromagnetic wave produced by the UE transmitting antenna and impinging on a receive antenna of the BS array can be \textit{locally} approximated by a plane wave. 
This approximation can be made provided that the distance between the transmit and receive antennas is greater than $2 L_r^2 / \lambda$, where $L_r$ represents the maximum dimension of the \textit{receive} antenna. 
\end{assumption}

Let ${\bf E}_{{\rm inc},m}$ denote the electric field incident on the $m$th antenna of the BS array, produced by the UE's voltage source. Based on Assumption 2, we model ${\bf E}_{{\rm inc},m}$ as a plane wave that reaches the receive antenna from a particular azimuth angle $\phi_{{\rm inc},m} \in [-\pi/2, \pi/2)$ and elevation angle $\theta_{{\rm inc},m} \in [-\pi/2, \pi/2)$. Assuming that the BS array consists of canonical minimum scattering (CMS) antennas and the incident field is linearly polarized, the $m$th element of ${\bf v}_{{\rm OC}}$ reads~\cite[Eq. (2-93)]{balanis}
\begin{equation}
\label{voc_einc}
v_{{\rm OC},m} = {\bf E}_{{\rm inc},m} \cdot {\bf l}^{(\rm r)}_{{\rm eff},m}(\theta_{{\rm inc},m},\phi_{{\rm inc},m})
\end{equation}
where ${\bf l}^{(\rm r)}_{{\rm eff},m}(\theta_{{\rm inc},m},\phi_{{\rm inc},m})$ is the \textit{effective length} ~\cite[Eq. (2-91)]{balanis} of the \textit{isolated} $m$th element of the array towards the $(\theta_{{\rm inc},m},\phi_{{\rm inc},m})$ direction. 
From Assumption 1, we have that~\cite[Eq. (2-92)]{balanis}
\begin{equation}
\label{eq:Einc}
\begin{split}
{\bf E}_{{\rm inc},m} = -\imagunit k  Z_{0} i_{\rm AT} \dfrac{e^{- \imagunit \frac{2\pi}{\lambda} r_m}}{4 \pi r_m} {\bf l}^{(\rm t)}_{{\rm eff}} (\theta_{{\rm inc},m},\phi_{{\rm inc},m})  
\end{split}
\end{equation}
where $i_{\rm AT}$ is the current feeding the antenna, $r_m$ is the distance between the centers of the transmit and receive antennas, ${\bf l}^{(\rm t)}_{{\rm eff}} (\theta_{{\rm inc},m},\phi_{{\rm inc},m})$ is the effective length of the transmit antenna in the direction of departure which, in a LoS scenario, coincides with the direction of arrival $(\theta_{{\rm inc},m},\phi_{{\rm inc},m})$. Plugging \eqref{eq:Einc} into \eqref{voc_einc} yields
\begin{equation}
\label{voc_einc_2}
\begin{split}
v_{{\rm OC},m} =  \alpha'(\theta_m,\phi_m,r_m)  Z_{0} i_{\rm AT} e^{- \imagunit \frac{2\pi}{\lambda} r_m}
\end{split}
\end{equation}
with
\begin{equation}
\label{alfap}
\alpha'(\theta_m,\phi_m,r_m) =  - \imagunit \dfrac{{\bf l}^{(\rm t)}_{{\rm eff}} (\theta_m, \phi_m) \cdot {\bf l}^{(\rm r)}_{{\rm eff}}(\theta_m,\phi_m)}{2 \lambda r_m} 
\end{equation}
where the term ${\bf l}^{(\rm t)}_{{\rm eff}} (\theta_m, \phi_m)\cdot{\bf l}^{(\rm r)}_{{\rm eff}}(\theta_m,\phi_m)$ accounts for the polarization loss~\cite[Sect. 2.12.2]{balanis}. For the sake of notation, we have dropped the subscript $_{\rm inc}$ so that $\theta_{{\rm inc},m}$ and $\phi_{{\rm inc},m}$ become $\theta_{m}$ and $\phi_{m}$, respectively. According to \eqref{voc_einc_2}, we can write ${\bf v}_{\rm OC} = Z_{0} i_{\rm AT} \boldsymbol{\alpha}'(\boldsymbol{\psi},\bf{r})\odot{\bf a}({\bf r})$. Finally, from ${\bf v}_{{\rm OC}} = v_{\rm G} {\bf d}_{{\rm OC}}$ and $i_{\rm AT} = F_{\rm T} (Z_{\rm G}+Z_{\rm T})^{-1} v_{\rm G}$, we obtain \eqref{eq:d_OC} where $\boldsymbol{\alpha}(\boldsymbol{\psi},{\bf r}) = F_{\rm T} (Z_{\rm G}+Z_{\rm T})^{-1} Z_0\boldsymbol{\alpha}'(\boldsymbol{\psi},{\bf r})$.

Depending on the relationship between the size of the array and its distance from the transmitting antenna, the expression of $v_{{\rm OC},m}$ can be simplified according to, for example, the \textit{Fresnel approximation} or the well-known \textit{planar wave approximation} \cite{Friedlander2019}. The latter differs from the planar wave approximation of Assumption 2 because it is relevant to the array while the one in Assumption 2 is relevant to the single array element. Under the planar wave approximation, 
\eqref{eq:d_OC} is reduced to the well-known expression in \eqref{doc}
where $\alpha(\theta,\phi,r) = F_{\rm T} (Z_{\rm G}+Z_{\rm T})^{-1} Z_0\alpha'(\theta,\phi,r)$, 
\begin{equation}
\alpha'(\theta,\phi, r) =  - \imagunit e^{\imagunit \psi_0}\dfrac{{\bf l}^{(\rm t)}_{{\rm eff}} (\theta, \phi) \cdot {\bf l}^{(\rm r)}_{{\rm eff}}(\theta,\phi)}{2 \lambda r}
\end{equation}
and $\psi_0 = - 2 \pi r/\lambda$ being the reference phase at array center. 
}

\vspace{-0.4cm}
\section*{Appendix B \\ Connection with the scattering representation}
Instead of dealing with voltages and currents, incident and reflected \textit{power waves} can describe multiport systems (e.g., \cite{Wallace2004}). At port $n$ of a multiport network, we define the \textit{scattering parameters} $a_{n}$ and $b_{n}$, representing the complex envelopes of the inward-propagating (incident) and outward-propagating (reflected) power waves, respectively. They relate to the voltage and current, $v_{n}$ and $i_{n}$, measured at the same port, as~\cite{Kurokawa1965},\cite[Ch. 4]{Pozar}:
\begin{equation}
\label{An}
a_{n}=\dfrac{v_{n}+Z_{n}i_{n}}{2 \sqrt{\re(Z_{n})}} \quad \quad
b_{n}=\dfrac{v_{n}-Z_{n}^{\ast}i_{n}}{2 \sqrt{\re(Z_{n})}}
\end{equation}
where $Z_{n}$ is a chosen reference impedance used for computing the scattering parameters. The physical meaning of $a_{n}$ and $b_{n}$ can be appreciated by computing $|a_{n}|^{2}-|b_{n}|^{2}=\re(v_{n}i^{\ast}_{n})$. 
which represents the total power flowing into port $n$. This is valid for any reference impedance $Z_{n}$. Hence, the total power flowing into a multiport system is $\re({\bf v}^{\Htran}{\bf i})=\|{\bf a}\|^{2}-\|{\bf b}\|^{2}$
where ${\bf v}$ and ${\bf i}$ are the vectors of the voltages and currents at the ports of the network, while $\bf a$ and $\bf b$ are vectors collecting the scattering parameters $a_{n}$ and $b_{n}$, respectively. The amplitudes of the incident and reflected waves are such that ${\bf b} = {\bf S} {\bf a}$ where $\bf S$ is the \textit{scattering} matrix.
The latter can be obtained from the impedance matrix $\bf Z$ as, e.g.,~\cite[Ch.~4, Eq. (4.68)]{Pozar} 
\begin{equation}
\label{SvsZ}
{\bf S}={\bf F} ({\bf Z-G^{\ast}}) ({\bf Z+G})^{-1}{\bf F}^{-1}
\end{equation}
where $\bf F$ and $\bf G$ are diagonal matrices with the $n$th diagonal elements $1/2\sqrt{\re(Z_{n})}$ and $Z_{n}$, respectively. By substituting each impedance matrix with its corresponding scattering matrix, based on~\eqref{SvsZ}, we describe the system in terms of scattering parameters instead of voltages and currents. Both descriptions are equivalent. For CMS antennas, the impedance description is preferred because it can be obtained directly from the isolated radiation pattern.

\bibliographystyle{IEEEtran}
\bibliography{IEEEabrv,refs}

\begin{thebibliography}{10}
\providecommand{\url}[1]{#1}
\csname url@samestyle\endcsname
\providecommand{\newblock}{\relax}
\providecommand{\bibinfo}[2]{#2}
\providecommand{\BIBentrySTDinterwordspacing}{\spaceskip=0pt\relax}
\providecommand{\BIBentryALTinterwordstretchfactor}{4}
\providecommand{\BIBentryALTinterwordspacing}{\spaceskip=\fontdimen2\font plus
\BIBentryALTinterwordstretchfactor\fontdimen3\font minus
  \fontdimen4\font\relax}
\providecommand{\BIBforeignlanguage}[2]{{%
\expandafter\ifx\csname l@#1\endcsname\relax
\typeout{** WARNING: IEEEtran.bst: No hyphenation pattern has been}%
\typeout{** loaded for the language `#1'. Using the pattern for}%
\typeout{** the default language instead.}%
\else
\language=\csname l@#1\endcsname
\fi
#2}}
\providecommand{\BIBdecl}{\relax}
\BIBdecl

\bibitem{marzetta2010noncooperative}
T.~L. Marzetta, ``Noncooperative cellular wireless with unlimited numbers of
  base station antennas,'' \emph{IEEE Trans. Wireless Commun.}, vol.~9, no.~11,
  pp. 3590--3600, Nov. 2010.

\bibitem{massivemimobook}
E.~Bj\"{o}rnson, J.~Hoydis, and L.~Sanguinetti, ``Massive {MIMO} networks:
  {Spectral}, energy, and hardware efficiency,'' \emph{Foundations and
  Trends{\textregistered} in Signal Processing}, vol.~11, no. 3-4, pp.
  154--655, 2017.

\bibitem{sanguinettiTCOM2020}
L.~Sanguinetti, E.~Bj\"ornson, and J.~Hoydis, ``Toward massive {MIMO} 2.0:
  Understanding spatial correlation, interference suppression, and pilot
  contamination,'' \emph{IEEE Trans. Commun.}, vol.~68, no.~1, pp. 232--257,
  Jan. 2020.

\bibitem{BJORNSON20193}
E.~Bj{\"o}rnson, L.~Sanguinetti, H.~Wymeersch, J.~Hoydis, and T.~L. Marzetta,
  ``Massive {MIMO} is a reality - what is next?: Five promising research
  directions for antenna arrays,'' \emph{Digital Signal Processing}, vol.~94,
  pp. 3 -- 20, Nov. 2019.

\bibitem{Huang2020}
C.~Huang and \textit{et al.}, ``Holographic {MIMO} surfaces for {6G} wireless
  networks: Opportunities, challenges, and trends,'' \emph{IEEE Wireless
  Commun.}, vol.~27, no.~5, pp. 118--125, Oct. 2020.

\bibitem{Rusek2018}
S.~Hu, F.~Rusek, and O.~Edfors, ``Beyond {Massive MIMO}: The potential of data
  transmission with large intelligent surfaces,'' \emph{IEEE Trans. Signal
  Proc.}, vol.~66, no.~10, pp. 2746--2758, May 2018.

\bibitem{balanis}
C.~A. Balanis, \emph{Antenna Theory: Analysis and Design}, 3rd~ed.\hskip 1em
  plus 0.5em minus 0.4em\relax John Wiley \& Sons, Inc., Hoboken, New Jersey,
  2005.

\bibitem{7831497}
A.~Li and C.~Masouros, ``Exploiting constructive mutual coupling in {P2P MIMO}
  by analog-digital phase alignment,'' \emph{IEEE Trans. Wireless Commun.},
  vol.~16, no.~3, pp. 1948--1962, March 2017.

\bibitem{6843218}
C.~Masouros, J.~Chen, K.~Tong, M.~Sellathurai, and T.~Ratnarajah, ``Exploiting
  transmit correlation and mutual coupling in {MIMO} transmitters,'' in
  \emph{European Wireless Conference}, 2014, pp. 1--6.

\bibitem{8350292}
X.~Chen, S.~Zhang, and Q.~Li, ``A review of mutual coupling in {MIMO}
  systems,'' \emph{IEEE Access}, vol.~6, April 2018.

\bibitem{9048753}
T.~L. Marzetta, ``Super-directive antenna arrays: Fundamentals and new
  perspectives,'' in \emph{Asilomar Conference on Signals, Systems, and
  Computers}, 2019, pp. 1--4.

\bibitem{9838533}
L.~Han, H.~Yin, and T.~L. Marzetta, ``Coupling matrix-based beamforming for
  superdirective antenna arrays,'' in \emph{IEEE Int. Conf. Commun.}, 2022, pp.
  5159--5164.

\bibitem{Heath_2023a}
N.~Deshpande, M.~R. Castellanos, S.~R. Khosravirad, J.~Du, H.~Viswanathan, and
  R.~W. H.~J. au2, ``A generalization of the achievable rate of a {MISO} system
  using bode-fano wideband matching theory,'' 2023.

\bibitem{10158708}
M.~Akrout, V.~Shyianov, F.~Bellili, A.~Mezghani, and R.~W. Heath,
  ``Super-wideband {Massive MIMO},'' \emph{IEEE J. Sel. Areas Commun.},
  vol.~41, no.~8, pp. 2414--2430, Aug. 2023.

\bibitem{Janaswamy2002}
R.~Janaswamy, ``Effect of element mutual coupling on the capacity of fixed
  length linear arrays,'' \emph{IEEE Antennas Wirel. Propag. Lett.}, vol.~1,
  pp. 157--160, March 2002.

\bibitem{Svantesson_ICASSP2001}
T.~Svantesson and A.~Ranheim, ``Mutual coupling effects on the capacity of
  multielement antenna systems,'' in \emph{IEEE International Conf. Acoustics,
  Speech, and Signal Processing}, vol.~4, 2001, pp. 2485--2488 vol. 4.

\bibitem{Wallace2004}
J.~Wallace and M.~Jensen, ``Mutual coupling in {MIMO} wireless systems: a
  rigorous network theory analysis,'' \emph{IEEE Trans. Wireless Commun.},
  vol.~3, no.~4, pp. 1317--1325, July 2004.

\bibitem{Nossek2010}
M.~T. Ivrlac and J.~A. Nossek, ``Toward a circuit theory of communication,''
  \emph{IEEE Trans. Circuits Syst. I: Regular Papers}, vol.~57, no.~7, pp.
  1663--1683, July 2010.

\bibitem{Nossek2014}
------, ``The multiport communication theory,'' \emph{IEEE Circuits and Systems
  Magazine}, vol.~14, no.~3, pp. 27--44, 2014.

\bibitem{yordanov2009arrays}
H.~Yordanov, M.~T. Ivrlac, P.~Russer, and J.~A. Nossek, ``Arrays of isotropic
  radiators-a field-theoretic justification,'' in \emph{Proc. ITG/IEEE Workshop
  on Smart Antennas}, 2009.

\bibitem{Ivrlac2009ICC}
M.~T. Ivrlac and J.~A. Nossek, ``Receive antenna gain of uniform linear arrays
  of isotrops,'' in \emph{IEEE Int. Conf. Commun.}, 2009, pp. 1--6.

\bibitem{Laas2020}
T.~Laas, J.~A. Nossek, and W.~Xu, ``Limits of transmit and receive array gain
  in massive {MIMO},'' in \emph{IEEE Wireless Communications and Networking
  Conference}, 2020, pp. 1--8.

\bibitem{ivrlavc2011diversity}
M.~T. Ivrla{\v{c}} and J.~A. Nossek, ``On the diversity performance of compact
  antenna arrays,'' in \emph{Proc. 30th Gen. Assembly Int. Union Radio
  Sci.(URSI)}, 2011, pp. 1--4.

\bibitem{Laas2020_Reciprocity}
T.~Laas, J.~A. Nossek, S.~Bazzi, and W.~Xu, ``On reciprocity in physically
  consistent {TDD} systems with coupled antennas,'' \emph{IEEE Trans. Wireless
  Commun.}, vol.~19, no.~10, pp. 6440--6453, Oct. 2020.

\bibitem{Bamelak_2023}
B.~Tadele, V.~Shyianov, F.~Bellili, and A.~Mezghani, ``Channel estimation with
  tightly-coupled antenna arrays,'' in \emph{IEEE International Conference on
  Acoustics, Speech and Signal Processing}, 2023, pp. 1--5.

\bibitem{Volodymyr_2022}
V.~Shyianov, M.~Akrout, F.~Bellili, A.~Mezghani, and R.~W. Heath, ``Achievable
  rate with antenna size constraint: {Shannon} meets {Chu and Bode},''
  \emph{IEEE Trans. Commun.}, vol.~70, no.~3, pp. 2010--2024, March 2022.

\bibitem{TseBook}
D.~Tse and P.~Viswanath, \emph{Fundamentals of Wireless Communication}.\hskip
  1em plus 0.5em minus 0.4em\relax Cambridge University Press, 2005.

\bibitem{Nyquist1928}
H.~Nyquist, ``Thermal agitation of electric charge in conductors,'' \emph{Phys.
  Rev.}, vol.~32, pp. 110--113, Jul 1928.

\bibitem{Rothe1956}
H.~Rothe and W.~Dahlke, ``Theory of noisy fourpoles,'' \emph{Proceedings of the
  IRE}, vol.~44, no.~6, pp. 811--818, 1956.

\bibitem{Kahn1965}
W.~Kahn and H.~Kurss, ``Minimum-scattering antennas,'' \emph{IEEE Trans.
  Antennas Propag.}, vol.~13, no.~5, pp. 671--675, Sept. 1965.

\bibitem{Warnick2009}
K.~F. Warnick, B.~Woestenburg, L.~Belostotski, and P.~Russer, ``Minimizing the
  noise penalty due to mutual coupling for a receiving array,'' \emph{IEEE
  Trans. Antennas Propag.}, vol.~57, no.~6, pp. 1634--1644, June 2009.

\bibitem{orfanidis}
S.~J. Orfanidis, \emph{Electromagnetic Waves and Antennas}.\hskip 1em plus
  0.5em minus 0.4em\relax Available:
  https://www.ece.rutgers.edu/~orfanidi/ewa/., 2016.

\bibitem{MullerISIT2012}
R.~R. M{\"u}ller, B.~E. Godana, M.~A. Sedaghat, and J.~B. Huber, ``On channel
  capacity of communication via antenna arrays with receiver noise matching,''
  in \emph{2012 IEEE Information Theory Workshop}, 2012, pp. 396--400.

\bibitem{Friedlander2020}
B.~Friedlander, ``The extended manifold for antenna arrays,'' \emph{IEEE Trans.
  Signal Process.}, vol.~68, pp. 493--502, Jan. 2020.

\bibitem{kraus1988antennas}
J.~Kraus, \emph{Antennas}, ser. Electrical engineering series.\hskip 1em plus
  0.5em minus 0.4em\relax McGraw-Hill, 1988.

\bibitem{Friedlander2019}
B.~Friedlander, ``Localization of signals in the near-field of an antenna
  array,'' \emph{IEEE Trans. Signal Process.}, vol.~67, no.~15, pp. 3885--3893,
  Aug. 2019.

\bibitem{Kurokawa1965}
K.~Kurokawa, ``Power waves and the scattering matrix,'' \emph{IEEE Trans.
  Microw. Theory Tech.}, vol.~13, no.~2, pp. 194--202, 1965.

\bibitem{Pozar}
D.~M. Pozar, \emph{{Microwave engineering; 4th ed.}}\hskip 1em plus 0.5em minus
  0.4em\relax Hoboken, NJ: Wiley, 2012.

\end{thebibliography}

\end{document}